\documentclass{llncs}

\usepackage{graphicx}
\usepackage{psfrag}

\usepackage{times}
\usepackage{enumerate}
\usepackage{bm}

\usepackage{amssymb}

\newcommand{\OPT}{\text{{\sc opt}}}
\newcommand{\opt}{\OPT}
\newcommand{\MPS}{{\rm MPS}} 
\newcommand{\MPSO}{{\rm MPS$_{\rm opt}$}} 
\newcommand{\eps}{\varepsilon}

\usepackage{amsmath,amsfonts,amssymb}
\usepackage{savesym}
\savesymbol{vec}
\usepackage{MnSymbol}
\usepackage{amsfonts}

\begin{document}

\title{Online Makespan Minimization with Parallel Schedules\vspace*{-0.3cm}}
\author{Susanne Albers \and Matthias Hellwig}
\institute{{Department of Computer Science, Humboldt-Universit\"at zu Berlin}
\email{\{albers,mhellwig\}@informatik.hu-berlin.de}\vspace*{-0.5cm}}

\maketitle

\begin{abstract}

Online makespan minimization is a classical problem in which a sequence of jobs $\sigma = J_1, \ldots, J_n$
has to be scheduled on $m$ identical parallel machines so as to minimize the maximum completion time of any job.
In this paper we investigate the problem with an essentially new model of resource augmentation.
More specifically, an online algorithm is allowed to build several schedules in parallel while processing $\sigma$. 
At the end of the scheduling process the best schedule is selected. This model can be viewed as 
providing an online algorithm with extra space, which is invested to maintain multiple solutions. The 
setting is of particular interest in parallel processing environments where each processor can maintain a 
single or a small set of solutions. 

As a main result we develop a $(4/3+\eps)$-competitive algorithm, for any $0<\eps\leq 1$, that uses a 
constant number of schedules. The constant is $1/\eps^{O(\log (1/\eps))}$. We also give 
a $(1+\eps)$-competitive algorithm, 
for any $0<\eps\leq 1$, that builds a polynomial number of $(m/\eps)^{O(\log (1/\eps) / \eps)}$ 
schedules. This value depends on $m$ but is independent of the input $\sigma$. The performance guarantees
are nearly best possible. We show that any algorithm that achieves a competitiveness smaller than $4/3$
must construct $\Omega(m)$ schedules. Our algorithms make use of novel guessing schemes that (1)~predict
the optimum makespan of a job sequence $\sigma$ to within a factor of $1+\eps$ and (2)~guess the job processing 
times and their frequencies in $\sigma$. In~(2) we have to sparsify the universe of all guesses so
as to reduce the number of schedules to a constant.

The competitive ratios achieved using parallel schedules are considerably smaller than
those in the standard problem without resource augmentation. Furthermore they are at least as good and
in most cases better than the ratios obtained with other means of resource augmentation for makespan
minimization.

\end{abstract}

\section{Introduction}
Makespan minimization is a fundamental and extensively studied problem in scheduling theory. Consider a
sequence of jobs $\sigma = J_1, \ldots, J_n$ that has to be scheduled on $m$ identical parallel
machines. Each job $J_t$ is specified by a processing time $p_t>0$, $1\leq t \leq n$. Preemption of 
jobs is not allowed. The goal is to minimize the makespan, i.\,e.\ the maximum completion time of
any job in the constructed schedule. We focus on the online version of the problem where the jobs
of $\sigma$ arrive one by one. Each incoming job $J_t$ has to be assigned immediately 
to one of the machines without knowledge of any future jobs $J_{t'}$, $t'>t$.

Online algorithms for makespan minimization have been studied since the 1960s. In an early paper
Graham~\cite{G} showed that the famous {\em List\/} scheduling algorithm is $(2-1/m)$-competitive.
The best online strategy currently known achieves a competitiveness of about 1.92. Makespan minimization 
has also been studied with various types of {\em resource augmentation\/}, giving an online algorithm  
additional information or power while processing $\sigma$. The following scenarios were considered.
(1)~An online algorithm knows the optimum makespan or the sum of the processing times of $\sigma$.
(2)~An online strategy has a buffer that can be used to reorder $\sigma$. Whenever a job arrives, it
is inserted into the buffer; then one job of the buffer is removed and placed in the current schedule. 
(3)~An online algorithm may migrate a certain number or volume of jobs. 

In this paper we investigate makespan minimization assuming that an online algorithm is allowed to
build several schedules in parallel while processing a job sequence $\sigma$. Each incoming job is
sequenced in each of the schedules. At the end of the scheduling process the best schedule is selected. 
We believe that this is a natural form of resource augmentation: In classical online makespan minimization, 
studied in the literature so far, an algorithm constructs a schedule while jobs arrive one by one. 
Once all jobs have arrived, the schedule may be executed. 
Hence in this standard framework there is a priori no reason why an algorithm should not be 
able to construct several solutions, the best of which is finally chosen. 

Our new proposed setting can be viewed
as providing an online algorithm with extra space, which is used to maintain several 
solutions. Very little is known about the value of 
extra space in the design of online algorithms. Makespan minimization with parallel schedules
is of particular interest in parallel processing environments where each processor can take care of
a single or a small set of schedules. We develop algorithms that require hardly any coordination 
or communication among the schedules. Last not least the proposed setting is interesting
w.\,r.\,t. to the foundations of scheduling theory, giving insight into the value of multiple
candidate solutions.

Makespan minimization with parallel schedules was also addressed by Kellerer et al.~\cite{KKST}. 
However, the paper focused on the restricted setting with  $m=2$ machines.
In this paper we explore the problem for a
general number $m$ of machines. As a main result we show that a constant number of schedules
suffices to achieve a significantly improved competitiveness, compared to the standard setting
without resource augmentation. The competitive ratios obtained are at least 
as good and in most cases better than those attained in the other models of resource augmentation 
mentioned above. 

The approach to grant an online algorithm extra space, invested to maintain multiple solutions,
could be interesting in other problems as well. The approach is viable in applications where an 
online algorithm constructs a solution that is used when the entire input has arrived.  
This is the case, for instance, in basic online graph coloring and matching problems~\cite{I,KVV,LST}. 
The approach is also promising in problems that can be solved by a set of independent agents, each of 
which constructs a separate solution. Good examples are online navigation and exploration problems
in robotics~\cite{BC,BRS,DKP}. Some results are known for graph search and 
exploration, see e.\,g.~\cite{BS,FGK,LS}, but the approach has not been studied for geometric 
environments.

\vspace*{0.1cm}

{\bf Problem definition:}
We investigate the problem {\em Makespan Minimization with Parallel Schedules (MPS)\/}. As always, the
jobs of a sequence $\sigma = J_1, \ldots, J_n$ arrive one by one and must be scheduled 
non-preemptively on $m$ identical parallel machines. Each job $J_t$ has a processing time $p_t >0$. 
In \MPS, an online algorithm ${\cal A}$ may maintain a set ${\cal S} = \{S_1, \ldots, S_l\}$ of schedules
during the scheduling process while jobs of $\sigma$ arrive. Each job $J_t$ is sequenced 
in each schedule $S_k$, $1\leq k \leq l$. At the end of $\sigma$, algorithm ${\cal A\/}$ selects a schedule 
$S_k\in {\cal S}$ having the smallest makespan and outputs this solution. The other schedules of ${\cal S\/}$
are deleted. 

As we shall show \MPS\ can be reduced to the problem variant where the optimum makespan of the job sequence
to the processed is known in advance. Hence let \MPSO\ denote the variant of \MPS\ where, prior to the
arrival of the first job, an algorithm ${\cal A\/}$ is given the value of the optimum makespan
$\opt(\sigma)$ for the incoming job sequence $\sigma$. 
An algorithm ${\cal A}$ for \MPS\ or \MPSO\ is 
$\rho$-competitive if, for every job sequence $\sigma$, it outputs a schedule whose makespan is at most 
$\rho$ times $\opt(\sigma)$.

\vspace*{0.1cm}
{\bf Our contribution:}
We present a comprehensive study of \MPS. We develop a $(4/3+\eps)$-competitive algorithm, for any
$0<\eps\leq 1$, using a constant number of $1/\eps^{O(\log (1/\eps))}$ schedules.
Furthermore, we give a $(1+\eps)$-competitive algorithm, for any $0<\eps\leq 1$, that uses a 
polynomial number of schedules. The number is $(m/\eps)^{O(\log (1/\eps) / \eps)}$, which
depends on $m$ but is independent of the job sequence $\sigma$. These performance guarantees are nearly 
best possible. The algorithms are obtained via some intermediate results that may be of 
independent interest.

First, in Section~\ref{sec:redu} we show that the original problem \MPS\ can be reduced to the variant
\MPSO\ in which the optimum makespan is known. More specifically, given any
$\rho$-competitive algorithm ${\cal A}$ for \MPSO\ we construct a $(\rho+\eps)$-competitive algorithm
${\cal A}^*(\eps)$, for any $0< \eps \leq 1$. If ${\cal A}$ uses $l$ schedules, then ${\cal A}^*(\eps)$
uses $l \cdot \lceil \log (1+ \frac{6\rho}{\eps}) / \log(1+\frac{\eps}{3\rho})\rceil$ schedules. 
The construction works for any algorithm ${\cal A}$ for \MPSO. In particular we could use a 1.6-competitive
algorithm by Chen et al.~\cite{CKK} that assumes that the optimum makespan is known and builds
a single schedule. We would obtain a $(1.6+\eps)$-competitive algorithm that builds at most 
$\lceil \log (1+ 10/\eps) / \log(1+\eps/5)\rceil$ schedules.

We proceed and develop algorithms for \MPSO. In Section~\ref{sec:ptas} we give a $(1+\eps)$-competitive
algorithm, for any $0<\eps\leq 1$, that uses 
$(\lfloor 2m/\eps\rfloor +1)^{\lceil \log(2/\eps) / \log(1+\eps/2) \rceil }$ schedules. In Section~\ref{sec:4/3}
we devise a $(4/3+\eps)$-competitive algorithm, for any $0<\eps\leq 1$, that uses 
$1/\eps^{O(\log (1/\eps))}$ schedules. Combining these algorithms with ${\cal A}^*(\eps)$, we
derive the two algorithms for \MPS\ mentioned in the above paragraph; see also Section~\ref{sec:mps}.
The number of schedules used by our strategies depends on $1/\eps$ and exponentially on 
$\log (1/\eps)$ or $1/\eps$. Such a dependence seems inherent if we wish to explore the full power
of parallel schedules. The trade-offs resemble those exhibited by PTASes in offline approximation.
Recall that the PTAS by Hochbaum and Shmoys~\cite{HS} for makespan minimization achieves
a $(1+\eps)$-approximation with a running time of $O((n/\eps)^{1/\eps^2})$.

In Section~\ref{sec:lb} we present lower bounds. We show that any online algorithm
for \MPS\ that achieves a competitive ratio smaller than 4/3 must construct more than $\lfloor m/3 \rfloor$
schedules. Hence the competitive ratio of 4/3 is best possible using a constant number of
schedules. We show a second lower bound that implies that the number of schedules of our 
$(1+\eps)$-competitive algorithm is nearly optimal, up to a polynomial factor.

Our algorithms make use of novel guessing schemes. ${\cal A}^*(\eps)$ works with guesses on the 
optimum makespan. Guessing and {\em doubling\/} the value of the optimal solution is a technique that
has been applied in other load balancing problems, see e.\,g.~\cite{Azar}. However here we 
design a refined scheme that carefully sets and readjusts guesses so that the resulting competitive
ratio increases by a factor of $1+\eps$ only, for any $\eps >0$. Moreover, the readjustment
and job assignment rules have to ensure that scheduling errors, made when guesses were to
small, are not critical. Our $(4/3+\eps)$-competitive algorithm works with guesses on the job processing
times and their frequencies in $\sigma$. In order to achieve a constant number of schedules, we have
to sparsify the set of all possible guesses. As far as we know such an approach has not been used 
in the literature before.

All our algorithms have the property that the parallel schedules are constructed basically independently.
The algorithms for \MPSO\ require no coordination at all among the schedules. In ${\cal A}^*(\eps)$ 
a schedule only has to report when it fails, i.\,e.\ when a guess on the optimum makespan is
too small.

The competitive ratios achieved with parallel schedules are considerably smaller than
the best ratios of about 1.92 known for the scenario without resource augmentation. Our ratio of $(4/3+\eps)$, 
for small $\eps$, is lower than the competitiveness of about 1.46 obtained in the settings where a 
reordering buffer of size $O(m)$ is available or $O(m)$ jobs may be reassigned. 
Skutella et al.~\cite{SSS} gave an online algorithm that is $(1+\eps)$-competitive if,
before the assignment of any job $J_t$, jobs of processing volume $2^{O((1/\eps)\log^2(1/\eps))}p_t$
may be migrated. Hence the total amount of extra resources used while scheduling $\sigma$
depends on the input sequence.


{\bf Related work:} Makespan minimization with parallel schedules was first studied by 
Kellerer et al.~\cite{KKST}. They assume that $m=2$ machines are available and two schedules may
be constructed. They show that in this case the optimal competitive ratio is 4/3.

We summarize results known for online makespan minimization without resource augmentation. As mentioned before, 
{\em List\/} is $(2-1/m)$-competitive. Deterministic online algorithms with a smaller competitive ratio were
presented in~\cite{A,BFKV,FW,GW,KPT}. The best algorithm currently known is 1.9201-competitive~\cite{FW}.
Lower bounds on the performance of deterministic strategies were given in~\cite{A,BKR,FKT,GRTW,R,RC}.
The best bound currently known is 1.88, see~\cite{R}. No randomized online algorithm whose competitive 
ratio is provably below the deterministic lower bound is currently known for general $m$. 

We next review the results for the various models of resource augmentation.
Articles~\cite{ANST,AST,AST2,AR,CKK,KKST} study makespan minimization assuming that an online algorithm knows
the optimum makespan or the sum of the processing times of $\sigma$. Chen et al.~\cite{CKK} developed a
1.6-competitive algorithm. Azar and Regev~\cite{AR} showed that no online algorithm can attain a competitive
ratio smaller than 4/3. The setting in which an online algorithm is given a reordering buffer
was explored in~\cite{EOW,KKST}. Englert et al.~\cite{EOW} presented an algorithm that, using a buffer
of size $O(m)$, achieves a competitive ratio of $W_{-1}(-1/e^2)/(1+ W_{-1}(-1/e^2))\approx 1.46$,
where $W_{-1}$ is the Lambert $W$ function. No algorithm using a buffer of size $o(n)$ can beat 
this ratio. 

Makespan minimization with job migration was addressed in~\cite{AH,SSS}. An algorithm that
achieves again a competitiveness of $W_{-1}(-1/e^2)/(1+ W_{-1}(-1/e^2))\approx 1.46$ and uses $O(m)$ job
reassignments was devised in~\cite{AH}. No algorithm using $o(n)$ reassignments can obtain
a smaller competitiveness. Sanders et al.~\cite{SSS} study a  scenario in which before the assignment 
of each job $J_t$, jobs up to a total processing volume of $\beta p_i$ may be migrated, for some constant
$\beta$. For $\beta=4/3$, they present a 1.5-competitive algorithm. They also show a $(1+\eps)$-competitive
algorithm, for any $\eps >0$, where $\beta = 2^{O((1/\eps)\log^2(1/\eps))}$. 

As for memory in online algorithms, Sleator and Tarjan~\cite{ST} studied the paging problem assuming that an
online algorithm has a larger fast memory than an offline strategy. Raghavan and Snir~\cite{RS} traded memory for randomness 
in online caching. 

{\bf Notation:} Throughout this paper it will be convenient to associate schedules with algorithms, i.\,e.\ a 
schedule $S_k$ is maintained by an algorithm $A_k$ that specifies how to assign jobs to machines in $S_k$. Thus an
algorithm ${\cal A}$ for \MPS\ or \MPSO\ can be viewed as a family $\{A_k\}_{k\in {\cal K}}$ of algorithms 
that maintain the various schedules. We will write ${\cal A} = \{A_k\}_{k\in {\cal K}}$. If ${\cal A}$ is an 
algorithm for \MPSO, then the value $\opt(\sigma)$ is of course given to all algorithms of $\{A_k\}_{k\in {\cal K}}$. 
Furthermore, the {\em load\/} of a machine always denotes the sum of the processing
times of the jobs already assigned to that machine.

\section{Reducing \MPS\ to \MPSO}\label{sec:redu}

In this section we will show that any $\rho$-competitive algorithm ${\cal A}$ for \MPSO\ can 
be used to construct a $(\rho+\eps)$-competitive algorithm ${\cal A^*}(\eps)$ for \MPS, for any $0<\eps \leq 1$. 
The main idea is to repeatedly execute ${\cal A}$ for a set of guesses on the optimum makespan.
The initial guesses are small and are increased whenever a guess turns out to be smaller than $\opt(\sigma)$. 
The increments are done in small steps so that, among the final guesses, there exists one that is upper bounded
by approximately $(1+\eps)\opt(\sigma)$. In the analysis of this scheme we will have to bound machine
loads caused by scheduling ``errors'' made when guesses were too small.
Unfortunately the execution of ${\cal A}$, given a guess $\gamma \neq \opt(\sigma)$, can lead to undefined
algorithmic behavior. As we shall show,
guesses $\gamma \geq \opt(\sigma)$ are not critical. However, guesses $\gamma < \opt(\sigma)$ have to be handled 
carefully. 

So let ${\cal A} = \{A_k\}_{k\in {\cal K}}$ be a $\rho$-competitive algorithm for \MPSO\ that, given guess 
$\gamma$, is executed on a job sequence $\sigma$. Upon the arrival of a job $J_t$, an algorithm
$A_k\in {\cal A}$ may {\em fail\/}  because the scheduling rules of $A_k$ do not specify a machine where 
to place $J_t$ in the current schedule $S_k$. We define two further conditions when an algorithm $A_k$
fails. The first one identifies situations where a makespan of $\rho\gamma$ is not preserved and hence
$\rho$-competitiveness may not be guaranteed. More precisely, $A_k$ would assign $J_t$ to a machine $M_j$ such 
that $\ell(j) + p_t > \rho\gamma$, where $\ell(j)$ denotes $M_j$'s machine load before the assignment.
The second condition identifies situations where $\gamma$ is not consistent with lower bounds on the
optimum makespan, i.\,e.\ $\gamma$ is smaller than the average machine load or the processing time
of $J_t$. Formally, an algorithm $A_k$ {\em fails\/} if a job $J_t$, $1\leq t \leq n$, has to be
scheduled and one of the following conditions holds.
\begin{enumerate}[(i)] 
 \item $A_k$ does not specify a machine where to place $J_t$ in the current schedule $S_k$.
 \item There holds $\ell(j)+p_t > \rho \gamma$, for the machine $M_j$ to which $A_k$ would assign $J_t$ 
      in $S_k$.
 \item There holds $\gamma < \sum_{t'\leq t} p_{t'}/m$ or $\gamma < p_t$.
\end{enumerate}

We first show that guesses $\gamma \geq\opt(\sigma)$ are not problematic. If a $\rho$-competitive 
algorithm ${\cal A} = \{A_k\}_{k\in {\cal K}}$ for MPS$_{\rm opt}$ is given a guess $\gamma \geq \opt(\sigma)$, 
then there exists an algorithm $A_k\in {\cal A}$ that does not fail during 
the processing of $\sigma$ and generates a schedule whose makespan is at most $\rho\gamma$. This is shown by the next lemma.

\begin{lemma}\label{lem:guess1}
Let ${\cal A} = \{A_k\}_{k\in {\cal K}}$ be a $\rho$-competitive algorithm for MPS$_{\rm opt}$ that, given guess
$\gamma$, is executed on a job sequence $\sigma$ with $\gamma \geq \opt(\sigma)$. Then there exists an
algorithm $A_k\in {\cal A}$ that does not fail during the processing of $\sigma$ and generates a schedule
whose makespan is at most $\rho\gamma$.
\end{lemma}
\begin{proof}
Let $S_{\rm opt}$ be an optimal schedule for the job sequence \linebreak $\sigma = J_1, \ldots, J_n$. Moreover, let $\ell(j)$ 
denote the load of machine $M_j$ in $S_{\rm opt}$, $1\leq j\leq m$. For any $j$ with $\ell(j) <\gamma$, define a
job $J'_j$ of processing time $p'_j = \gamma-\ell(j)$. Let $\sigma'$ be the job sequence consisting
of $\sigma$ followed by the new jobs $J'_j$. These up to $m$ jobs may be appended to $\sigma$ in any order.
Obviously $\opt(\sigma') = \gamma$. Hence when ${\cal A\/}$ using guess $\gamma$ is executed on $\sigma'$,
there must exist an algorithm $A_{k^*}\in {\cal A}$ that generates a schedule with a makespan of at most 
$\rho\gamma$. Since $\sigma$ is a prefix of $\sigma'$, this algorithm $A_{k^*}$ does not fail and generates
a schedule with a makespan of at most $\rho\gamma$, when ${\cal A}$ given guess $\gamma$ is executed on
$\sigma$. \hspace*{\fill}{$\Box$}
\end{proof}

\vspace*{0.1cm}

{\bf Algorithm for MPS:} We describe our algorithm ${\cal A^*}(\eps,h)$ for \MPS, where 
$0<\eps\leq 1$ and $h\in \mathbb{N}$ may be
chosen arbitrarily. The construction takes as input any algorithm ${\cal A} = \{A_k\}_{k\in {\cal K}}$
for \MPSO. For a proper choice of $h$, ${\cal A^*}(\eps,h)$ will be $(\rho+\eps)$-competitive, provided
that ${\cal A}$ is $\rho$-competitive. 

At any time ${\cal A^*}(\eps,h)$ works with $h$ guesses $\gamma_1 < \ldots < \gamma_h$ on the optimum
makespan for the incoming job sequence $\sigma$. These guesses may be adjusted during the processing
of $\sigma$; the update procedure will be described in detail below. For each guess $\gamma_i$, $1\leq i \leq h$,
${\cal A^*}(\eps,h)$ executes ${\cal A}$. Hence ${\cal A^*}(\eps,h)$ maintains a total of $h|{\cal K}|$ schedules,
which can be partitioned into subsets ${\cal S}_1, \ldots, {\cal S}_h$. Subset ${\cal S}_i$ contains
those schedules generated by ${\cal A}$ using $\gamma_i$, $1\leq i \leq h$. Let $S_{ik}\in {\cal S}_i$
denote the schedule generated by $A_k$ using $\gamma_i$. 

A job sequence $\sigma$ is processed as follows. Initially, upon the arrival of the first job $J_1$, the
guesses are initialized as $\gamma_1 = p_1$ and $\gamma_i = (1+\eps)\gamma_{i-1}$, for $i=2, \ldots, h$. 
Each job $J_t$, $1\leq t \leq n$, is handled in the following way. Of course each such job is sequenced
in every schedule $S_{ik}$, $1\leq i \leq h$ and $1\leq k \leq |{\cal K}|$. Algorithm  ${\cal A^*}(\eps,h)$
checks if $A_k$ using $\gamma_i$ fails when having to sequence $J_t$ in $S_{ik}$. We remark that this
check can be performed easily by just verifying if one of the conditions (i--iii) holds. If $A_k$ using 
$\gamma_i$ does not fail and has not failed since the last adjustment of $\gamma_i$, then in $S_{ik}$
job $J_t$ is assigned to the machine specified by $A_k$ using $\gamma_i$. The initialization of a
guess is also regarded as an adjustment.  If $A_k$ using $\gamma_i$ does fail, then
$J_t$ and all future jobs are always assigned to a least loaded machine in $S_{ik}$
until $\gamma_i$ is adjusted the next time. 

Suppose that after the sequencing of $J_t$ all algorithms of ${\cal A} = \{A_k\}_{k\in {\cal K}}$ using
a particular guess $\gamma_i$ have failed since the last adjustment of this guess. Let $i^*$ be the
largest index $i$ with this property. Then the guesses $\gamma_1, \ldots, \gamma_{i^*}$ are adjusted.
Set $\gamma_1 = (1+\eps)\max\{\gamma_h,p_t, \sum_{1\leq t'\leq t} p_{t'}/m\}$ and 
$\gamma_i = (1+\eps)\gamma_{i-1}$, for $i=2, \ldots, i^*$. For any readjusted  guess $\gamma_i$, 
$1\leq i \leq i^*$, algorithm ${\cal A}$ using $\gamma_i$ ignores all jobs $J_{t'}$ with $t'<t$ when
processing future jobs of $\sigma$. Specifically, when making scheduling decisions and determining
machine loads, algorithm $A_k$ using $\gamma_i$ ignores all job $J_{t'}$ with $t'<t$ in its schedule
$S_{ik}$. These jobs are also ignored when ${\cal A^*}(\eps,h)$ checks if $A_k$ using guess $\gamma_i$ fails
on the arrival of a job. Furthermore, after the assignment of $J_t$, machines in $S_{ik}$ machines are renumbered
so that $J_t$ is located on a machine it would occupy if it were the first job of an input sequence. 

When guesses have been adjusted, they are renumbered, together with the corresponding schedule sets ${\cal S}_i$, 
such that again $\gamma_1 < \ldots < \gamma_h$. Hence at any time $\gamma_1 = \min_{1\leq i\leq h} \gamma_i$ and
$\gamma_i \geq (1+\eps)\gamma_{i-1}$, for $i=2, \ldots, h$. We also observe that whenever a guess is adjusted, its 
value increases by a factor of at least $(1+\eps)^h$. A summary of ${\cal A^*}(\eps,h)$ is given in Figure~\ref{fig:1}.

\begin{figure}[h]
\fbox{
\begin{minipage}{11.7cm}
{\bf Algorithm ${\cal A^*}(\eps,h)$} \\[-15pt]
\begin{itemize}
\item[1.] Set $\gamma_i = p_1(1+\eps)^{i-1}$, for $i= 1, \ldots, h$.\\[-8pt]
\item[2.] At time $t$ execute the following steps.\\[-8pt]
\begin{itemize}
\item[(a)] $J_t$ is sequenced as follows in each $S_{ik}$. If $A_k$ using $\gamma_i$ fails or has
 failed since the last adjustment of $\gamma_i$, then assign $J_t$ to a least loaded machine. Otherwise
  assign it to the machine specified by $A_k$, ignoring jobs that arrived before the last adjustment of
 $\gamma_i$.\\[-8pt]
\item[(b)] If all algorithms $\{A_k\}_{k\in {\cal K}}$ for some $\gamma_{i}$ have failed since the last readjustment
  of $\gamma_{i}$, then let $i^*$ be the largest index with this property. Set $\gamma_i = (1+\eps)^i \max\{\gamma_h, p_t, \sum_{t'\leq t} p_{t'}/m\}$,  for $i= 1, \ldots, i^*$. Renumber the guesses such that $\gamma_1 < \ldots < \gamma_h$.\\[-5pt]
\end{itemize}
\end{itemize}
\end{minipage} 
}
\caption{The algorithm ${\cal A^*}(\eps,h)$}\label{fig:1}
\end{figure}

We obtain the following theorem.

\begin{theorem}\label{th:guess1}
Let ${\cal A}= \{A_k\}_{k\in {\cal K}}$ be a $\rho$-competitive algorithm for MPS$_{\rm opt}$.
Then for any $0<\eps \leq 1$ and 
$h = \lceil \log (1+ \frac{6\rho}{\eps}) / \log(1+\frac{\eps}{3\rho})\rceil$, algorithm
${\cal A^*}(\eps) = {\cal A^*}(\eps/(3\rho),h)$ for MPS is $(\rho+\eps)$-competitive and uses 
$h |{\cal K}|$ schedules. 
\end{theorem}

For the analysis of ${\cal A^*}(\eps,h)$ we need the following lemma. 

\begin{lemma}\label{lem:guess2}
After ${\cal A^*}(\eps,h)$ has processed a job sequence $\sigma$, there holds $\gamma_1\leq (1+\eps)\opt(\sigma)$.
\end{lemma}
\begin{proof}
At any time ${\cal A^*}(\eps,h)$ maintains $h$ guesses. We can view these guesses as being
stored in $h$ variables. A variable is updated whenever its current guess is increased. Hence
during the processing of $\sigma$ a variable may take any position in the sorted sequence of
guesses. We analyze the steps in which ${\cal A^*}(\eps,h)$ adjusts guesses.

We first show that when ${\cal A^*}(\eps,h)$ adjusts a guess $\gamma$, then $\gamma < \opt(\sigma)$. So
suppose that after the arrival of a job $J_t$, ${\cal A^*}(\eps,h)$ adjust guesses $\gamma_1, \ldots,
\gamma_{i^*}$, where $i^*$ is the largest index $i$ such that all algorithms $\{A_k\}_{k\in {\cal K}}$ 
using $\gamma_i$ have failed. We prove $\gamma_{i^*} < \opt(\sigma)$, which implies the desired statement
because guesses are numbered in order of increasing value. Let $t^*$, with $t^* <t$, be the most recent time when 
the variable storing $\gamma_{i^*}$ was updated last. If the variable has never been updated since its
initialization, then let $t^*=1$. All the algorithms $\{A_k\}_{k\in {\cal K}}$ using $\gamma_{i^*}$ ignore
the jobs having arrived before $J_{t^*}$ when making scheduling decisions for $J_{t^*}, \ldots, J_t$. 
Let $\sigma^* = J_{t^*}, \ldots, J_t$. There holds, $\opt(\sigma^*) \leq \opt(\sigma)$. 
If $\gamma_{i^*} \geq \opt(\sigma)$ held true, then by Lemma~\ref{lem:guess1} there would be an algorithm 
$A_{k^*}\in \{A_k\}_{k\in {\cal K}}$ that, using guess $\gamma_{i^*}$, does not fail when handling $\sigma^*$.
This contradicts the fact that at time $t$ all algorithms $\{A_k\}_{k\in {\cal K}}$ using
$\gamma_{i^*}$ fail or have failed since the arrival of $J_{t^*}$. 

Let $\gamma_1^e$ denote the value of the smallest guess when ${\cal A^*}(\eps,h)$  has finished processing
$\sigma$. We distinguish two cases depending on whether or not the variable storing $\gamma_1^e$ has
ever been updated since its initialization. If the variable has never been updated, then
$\gamma_1^e= p_1(1+\eps)^{i-1}$, for some $i\in \{1,\ldots, h\}$. If $i=1$, there is nothing to show
because $p_1\leq \opt(\sigma)$. If $i>1$, then the initial guess of value $\gamma_{i-1} = p_1(1+\eps)^{i-2}$
must have been adjusted. This implies, as shown above, $\gamma_{i-1} < \opt(\sigma)$ and the lemma follows 
because $\gamma_1^e = (1+\eps) \gamma_{i-1}$. 

In the remainder of the proof we assume that the variable $g$ storing $\gamma_1^e$ has been updated. Consider 
the last update of $g$  before the end of $\sigma$ and suppose that it took place on 
the arrival of job $J_{t^*}$. First assume that $g$ stores the smallest guess, among the $h$ guesses, before 
the update. Then $\gamma_1^e = (1+\eps)\max\{\gamma^*,p_{t^*}, \sum_{1\leq t'\leq t^*} p_{t'}/m\}$, where
$\gamma^*$ is the largest guess before the update. If $\gamma^*$ is also adjusted on the arrival of $J_{t^*}$,
then we are done because, as shown above, $\gamma^* < \opt(\sigma)$ and thus
$\max\{\gamma^*,p_{t^*}, \sum_{1\leq t'\leq t^*} p_{t'}/m\} \leq \opt(\sigma)$. If $\gamma^*$ is
not adjusted on the arrival of $J_{t^*}$, then $\gamma_1^e$ is the smallest guess greater than $\gamma^*$ 
after the update. By the end of $\sigma$ guess $\gamma^*$ must be adjusted since otherwise $\gamma_1^e$ cannot
become the smallest guess. Again $\gamma^* < \opt(\sigma)$ and we are done.

Finally assume that before the update $g$ does not store the smallest guess. Let $g'$ be the variable
that stores the largest guess smaller than that in $g$. After the update there holds $\gamma_1^e = (1+\eps)\gamma$,
where $\gamma$ is the guess stored in $g'$ after the update. Until the end of $\sigma$, $\gamma$ must be
adjusted again since otherwise $\gamma_1^e$ cannot become the smallest guess. Again $\gamma <\opt(\sigma)$ and
hence $\gamma_1^e < (1+\eps)\opt(\sigma)$. \hspace*{\fill}{$\Box$}
\end{proof}

\begin{proof}[of Theorem~\ref{th:guess1}]
Throughout the proof let $h = \lceil \log (1+ \frac{6\rho}{\eps}) / \log(1+\frac{\eps}{3\rho})\rceil$ and
${\cal A^*}(\eps) = {\cal A^*}(\eps/(3\rho),h)$. Consider an arbitrary job sequence and let $\gamma_1$ 
be the smallest of the $h$ guesses maintained by ${\cal A^*}(\eps)$ at the end of $\sigma$. Let
${\cal S}_1$ be the set of schedules associated with $\gamma_1$, i.\,e.\ ${\cal S}_1$ was generated
by ${\cal A}= \{A_k\}_{k\in {\cal K}}$ using a series of guesses ending with $\gamma_1$. Let
$\gamma(0) < \ldots < \gamma(s)$, with $s\geq 0$, be this series and $g$ be the variable
that stored these guesses. Here $\gamma(0)$ is one of the initial guesses and $\gamma(s)= \gamma_1$. 

A first observation is that at the end of $\sigma$ there exists an algorithm $A_{k^*}\in \{A_k\}_{k\in {\cal K}}$
that using $\gamma_1$ has not failed. This holds true if $g$ was set to $\gamma_1 = \gamma(s)$ upon
the arrival of a job $J_t$ with $t<n$ because the failure of all algorithms $\{A_k\}_{k\in {\cal K}}$
using $\gamma_1$ would have caused an adjustment of $\gamma_1$. This also holds true if $g$ was set to 
$\gamma_1$ upon the arrival of $J_n$ because in this case none of the algorithms $\{A_k\}_{k\in {\cal K}}$
using $\gamma_1$ has failed at the end of $\sigma$. So let $A_{k^*}\in \{A_k\}_{k\in {\cal K}}$ be
an algorithm that using $\gamma_1$ has not failed and let $S_{1k^*}$ be the associated schedule. 
We prove that the load of every machine in $S_{1k^*}$ is upper bounded by $(\rho+\eps)\opt(\sigma)$. 
This establishes the theorem.

Let $t_0=1$. If the variable $g$ was updated during the processing of $\sigma$, then let $t_1, \ldots, t_s$
be these points in time, i.\,e.\ the arrival of $J_{t_i}$ caused an update of $g$ and the variable was
set to $\gamma(i)$, $1\leq i \leq s$. For any machine $M_j$, $1\leq j \leq m$, in
$S_{1k^*}$ let $\ell(j)$ denote its final load at the end of $\sigma$. Moreover, let $\ell_{t_i}(j)$ 
denote its load due to jobs $J_t$ with $t\geq t_i$, for $i=0, \ldots, s$. Obviously
\begin{align}\label{eq:b1}
 \ell(j) =  \ell_{t_s}(j) +\sum_{i=0}^{s-1} \left ( \ell_{t_{i}}(j) - \ell_{t_{i+1}}(j) \right ).
\end{align}

We first show that $\ell_{t_s}(j) \leq \rho\gamma_1$. Immediately after $J_{t_s}$ has been scheduled
$M_j$'s load consisting of jobs $J_{t'}$ with $t'\geq t_s$ is at most $p_{t_s}$. Since $g$ was set to
$\gamma(s) = \gamma_1$ on the arrival of $J_{t_s}$, the guess adjustment rule ensures $p_{t_s} \leq \gamma_1$.
Until the end of $\sigma$ algorithm $A_{k^*}$ using $\gamma_1$ does not fail and hence condition~(ii)
specifying the failure of algorithms implies that the assignment of each further job does not create
a machine load greater than $\rho\gamma_1$ in $S_{1k^*}$. 

We next show $\ell_{t_{i}}(j) - \ell_{t_{i+1}}(j) \leq \max\{\rho, 2\} \gamma(i)$, for each $i=0,\dots,s-1$. The latter
difference is the load on machine $M_j$ caused by jobs of the subsequence $J_{t_i}, \ldots, J_{t_{i+1}-1}$. Hence it
suffices to show that after the assignment of any $J_t$, with $t_i \leq t < t_{i+1}$, $M_j$'s load due
to jobs $J_{t'}$, with $t'\geq t_i$, is at most $\max\{\rho,2\}\gamma(i)$. After the assignment of $J_{t_i}$
$M_j$'s respective load $\ell_{t_i}(j)$ is at most $p_{t_i}$ and this value is upper bounded by $\gamma(i)$ as ensured by the
guess adjustment rule. At times $t>t_i$, while $A_{k^*}$ using $\gamma(i)$ has not failed, $M_j$'s load
due to jobs $J_{t'}$ with $t'\geq t_i$ does not exceed $\rho\gamma(i)$ as ensured by condition~(ii) specifying
the failure of algorithms. Finally consider a time $t$, $t_i < t <t_{i+1}$, at which $A_{k^*}$
fails or has failed. The incoming job $J_t$ is assigned to a least loaded machine. Hence if $J_t$ is
placed on $M_j$, then the resulting machine load due to jobs $J_{t'}$ with $t'\geq t_i$ is upper bounded by 
$\sum_{t_i\leq t'< t}p_{t'}/m + p_t \leq \sum_{1\leq t'\leq t}p_{t'}/m + p_t$. Observe that after
the arrival of $J_t$ there exists an algorithm $A_k\in {\cal A}$ that using $\gamma(i)$ has not yet
failed, since otherwise $\gamma(i)$ would be adjusted before time $t_{i+1}$. Condition~(iii) defining
the failure of algorithms ensures that $\sum_{1\leq t'\leq t} p_{t'}/m \leq \gamma(i)$ and $p_t\leq \gamma(i)$.
We obtain that $M_j$'s machine load is at most $2\gamma(i)$.

We conclude that (\ref{eq:b1}) is upper bounded by 
\begin{align}
\rho \gamma_1 + \sum_{i=0}^{s-1}  \max\{\rho,2\}\gamma(i).  \label{eq:b2}
\end{align}
By Lemma~\ref{lem:guess2}, $\gamma_1 =\gamma(s) \leq (1+\eps/(3\rho))  \opt(\sigma)$. 
At the end of the description of ${\cal A^*}(\eps,h)$ we observed that whenever a guess is adjusted it
increases by a factor of at least $(1+\eps)^h$. Hence $\gamma(i) \geq (1+\eps/(3\rho))^h \gamma(i-1)$.
It follows that $\gamma(i) \leq \frac{\gamma(s)}{(1+(\eps/3\rho))^{(s-i)\cdot h}}$, for every $0 \leq i \leq s$. Hence~(\ref{eq:b2}) is upper bounded by
\begin{eqnarray}
\lefteqn{\rho(1+\frac{\eps}{3\rho})\opt(\sigma) + \sum_{i=0}^{s-1}  \frac{\max\{\rho,2\}\gamma(s)}{(1+\eps/(3\rho))^{h\cdot (s-i)}} \notag }\\
& \leq & \rho (1+\frac{\eps}{3\rho})\opt(\sigma) +  \rho(1+\frac{\eps}{3\rho})\opt(\sigma)\sum_{i=0}^{s-1}  \frac{2}{(1+\eps/(3\rho))^{h\cdot (s-i)}} \label{eq:xb2}\\
& \leq & \rho (1+\frac{\eps}{3\rho})\opt(\sigma) \left ( 1 +  \sum_{i=1}^{\infty}  \frac{2}{(1+\eps/(3\rho))^{h\cdot i}} \right ) \notag  \\
& =  & \rho (1+\frac{\eps}{3\rho})\opt(\sigma) \left ( 1 +    \frac{2}{(1+\eps/(3\rho))^h-1} \right )   \label{eq:b3} \\
&\leq  & \rho (1+\frac{\eps}{3\rho})^2\opt(\sigma) \ \leq \ \rho (1+\frac{\eps}{\rho})\opt(\sigma) = (\rho+\eps)\opt(\sigma).    \label{eq:b4} 
\end{eqnarray}
Here~(\ref{eq:xb2}) uses the fact that $\max\{\rho,2\}\leq 2\rho$ and, as mentioned above, is a consequence of
Lemma~\ref{lem:guess2}. Line~(\ref{eq:b3}) follows from the Geometric Series and, finally, (\ref{eq:b4}) is by the choice of $h$ and the assumption $0< \eps\leq 1$. \hspace*{\fill}{$\Box$}
\end{proof}

\section{A { $\mathbf{(1+\eps)}$}-competitive algorithm for \MPSO}\label{sec:ptas}
We present an algorithm ${\cal A}_1(\eps)$ for \MPSO\ that attains a competitive ratio of $1+\eps$, for any
$\eps>0$. The number of parallel schedules will be 
$(\lfloor 2m/\eps\rfloor +1)^{\lceil \log(2/\eps) / \log(1+\eps/2) \rceil }$.
The algorithms will yield a $(1+\eps)$-competitive strategy for $\MPS$ and, furthermore,
will be useful in the next section where we develop a $(4/3+\eps)$-competitive algorithm for \MPSO. 
There ${\cal A}_1(\eps)$ will be used as subroutine for a small, constant number of $m$.

{\bf Description of ${\cal A}_1(\eps)$:} Let $\eps>0$ be arbitrary. Recall that in \MPSO\ the optimum makespan
$\opt(\sigma)$ for the incoming job sequence $\sigma$ is initially known. Assume without loss of generality that
$\opt(\sigma)=1$. Then all job processing times are in $(0,1]$. Set $\eps' = \eps/2$. First we partition
the range of possible job processing times into intervals $I_0, \ldots, I_l$ such, within each interval $I_i$ 
with $i\geq 1$, the values differ by a factor of at most $1+\eps'$. Such a partitioning is 
standard and has been used e.\,g.\ in the PTAS for offline makespan minimization~\cite{HS}.
Let $l={\lceil \log(1/\eps') / \log(1+\eps') \rceil}$. Set $I_0 =(0,\eps']$ and 
$I_i = ((1+\eps')^{i-1}\eps', (1+\eps')^i\eps']$, for $i=1, \ldots, l$. Obviously 
$I_0\cup \ldots \cup I_l = (0, (1+\eps')^l\eps']$ and $(0,1] \subseteq (0, (1+\eps')^l\eps']$.
A job is {\em small\/} if its processing time is at most $\eps'$ and hence contained in $I_0$; otherwise the
job is {\em large\/}. 

Each job sequence $\sigma$ with $\opt(\sigma) = 1$ contains at most $\lfloor m/\eps'\rfloor$
large jobs. For each possible distribution of large jobs over the processing time intervals 
$I_1, \ldots, I_l$, algorithm ${\cal A}_1(\eps)$ prepares one algorithm/schedule.
Let $V=\{(v_1,\dots,v_l) \in \mathbb{N}^l_0 \mid  v_i \leq \lfloor m/\eps' \rfloor \}$. There holds 
$|V|=(\lfloor m/\eps'\rfloor+1)^l$. Let ${\cal A}_1(\eps) = \{A_v\}_{v\in V}$. For any vector 
$v=(v_1, \ldots, v_n)\in V$, algorithm $A_v$ works as follows. It assumes that the incoming job 
sequence $\sigma$ contains exactly $v_i$ jobs with a processing time in $I_i$, for $i=1, \ldots, l$. 
Moreover, it pessimistically assumes that each processing time in $I_i$ takes the largest possible
value $(1+\eps')^i\eps'$. Hence, initially $A_v$ computes an optimal schedule $S_v^*$ for a job 
sequence consisting of $v_i$ jobs with a processing time of $(1+\eps')^i\eps'$, for $i=1, \ldots, l$. Small
jobs are ignored. Since running time is not an issue in the design of online algorithms, such a
schedule $S_v^*$ can be computed exactly. Alternatively, an $(1+\eps')$-approximation to the
optimal schedule can be computed using the PTAS by Hochbaum and Shmoys~\cite{HS}.
Let $n_i^*(j)$ denote the number of jobs with a processing time of $(1+\eps')^i\eps'\in I_i$ assigned 
to machine $M_j$ in $S_v^*$, where $1\leq i \leq l$ and $1\leq j \leq m$. Moreover, let
$\ell^*(j) = \sum_{i=1}^l n_i^*(j) (1+\eps')^i\eps'$ be the load on machine $M_j$ in $S_v^*$, 
$1\leq j\leq m$. 

When processing the actual job sequence $\sigma$ and constructing a real schedule $S_v$, $A_v$ uses
$S_v^*$ as a guideline to make scheduling decisions. At any time during the scheduling process, let
$n_i(j)$ be the number of jobs with a processing time in $I_i$ that have already been assigned to
machine $M_j$ in $S_v$, where again $1\leq i \leq l$ and $1\leq j \leq m$. Each incoming job $J_t$,
$1\leq t\leq n$, is handled as follows. If $J_t$ is large, then let $I_i$ with $1\leq i \leq l$ be the 
interval such that $p_t\in I_i$. Algorithm $A_v$ checks if there is a machine $M_j$ such that 
$n_i^*(j) - n_i(j) >0$, i.\,e.\ there is a machine that can still accept a job with a processing time
in $I_i$ as suggested by the optimal schedule $S_v^*$. If such a machine $M_j$ exists, then $J_t$ is
assigned to it; otherwise $J_t$ is scheduled on an arbitrary machine. If $J_t$ is small, then $J_t$ is 
assigned to a machine $M_j$ with the smallest current value $\ell^*(j) + \ell_s(j)$. Here
$\ell_s(j)$ denotes the current load on machine $M_j$ caused by small jobs in $S_v$. 
A summary of ${\cal A}_1(\eps)$ is given in Figure~\ref{fig:2}. Subsequently we show Theorem~\ref{th:guess2}.

\begin{figure}[h]
\fbox{
\begin{minipage}{11.7cm}
{\bf Algorithm ${\cal A}_1(\eps)$} \\[-15pt]
\begin{itemize}
\item[1.] ${\cal A}_1(\eps) = \{A_v\}_{v\in V}$, where 
 $V=\{(v_1,\dots,v_l) \in \mathbb{N}^l_0 \mid  v_i \leq \lfloor m/\eps' \rfloor \}$\\
 with $\eps'= \eps/2$ and $l={\lceil \log(1/\eps') / \log(1+\eps') \rceil}$.\\[-8pt]
\item[2.] $A_v$ works as follows.\\[-8pt]
\begin{itemize}
\item[(a)] Compute optimal schedule $S^*_v$ for input consisting of $v_i$ jobs of processing time
$(1+\eps')^i\eps'$, $1\leq i \leq l$.\\[-8pt]
\item[(b)] In $S_v$ each $J_t$ is sequenced in the following way.\\
If $p_t >\eps'$, then determine $I_i$ such that $p_t\in I_i$. If $\exists\  M_j$ with 
 $n_i^*(j) - n_i(j)>0$, then assign $J_t$ to it; otherwise assign $J_t$ to an arbitrary machine.\\
If $p_t \leq \eps'$, then assign $J_t$ to $M_j$ with the smallest value $\ell^*(j) + \ell_s(j)$.\\[-8pt]
\end{itemize}
\end{itemize}
\end{minipage} 
}
\caption{The algorithm ${\cal A}_1(\eps)$}\label{fig:2}
\end{figure}

\begin{theorem}\label{th:guess2}
For any $\eps> 0$, ${\cal A}_1(\eps)$ is $(1+\eps)$-competitive and uses at most $(\lfloor 2m/\eps\rfloor +1)^{\lceil \log(2/\eps) / \log(1+\eps/2) \rceil }$ schedules.
\end{theorem}
\begin{proof}
The bound on the number of schedules simply follows from the fact that ${\cal A}_1(\eps)$ maintains
$|V|= (\lfloor m/\eps' \rfloor +1)^l$ schedules where $\eps'= \eps/2$ and 
$l={\lceil \log(1/\eps') / \log(1+\eps') \rceil}$.

Let $\sigma$ be an arbitrary job sequence and let $v_i$ be the number of jobs with a processing time in 
$I_i$, for $i=1, \ldots, l$. Since any $v_i$ is upper bounded by $\lfloor m/\eps' \rfloor$, the resulting
vector $v = (v_1, \ldots, v_l)$ is in $V$. For this vector $v$, consider the associated algorithm $A_v$.
We prove that when $A_v$ has finished processing $\sigma$, the resulting schedule $S_v$ has a makespan
of at most $(1+\eps) = (1+\eps)\opt(\sigma)$. Recall again that we assume without loss of generality
that $\opt(\sigma)=1$.

We analyze the steps in which $A_v$ assigns jobs $J_t$, $1\leq t\leq n$, to machines in $S_v$. If $J_t$ is
large with $p_t\in I_i$, $1\leq i\leq l$, then there must exist a machine $M_j$ in the current schedule
$S_v$ such that $n_i^*(j) - n_i(j)>0$. Algorithm $A_v$ will assign $J_t$ to such a machine. Hence after
the processing of $\sigma$, for any $M_j$ in $S_v$, the total load caused by large jobs is upper bounded
by $\ell^*(j)$. We next argue that this value is at most $(1+\eps')\opt(\sigma)$. Consider an optimal
schedule $S_{\rm opt}$ for $\sigma$. Modify this schedule by (a)~deleting all small jobs and (b)~rounding
each job processing time in $I_i$ to $(1+\eps')^i\eps'$, for $i=1,\ldots, l$. The resulting schedule 
schedule $S'_{\rm opt}$ has a makespan of at most $(1+\eps')\opt(\sigma)$. Furthermore $S'_{\rm opt}$ is
a schedule for an input sequence consisting of $v_i$ jobs of processing time $(1+\eps')^i\eps'$. Since
$S^*_v$ is an optimal schedule for this input, each machine load $\ell^*(j)$ is upper bounded by 
$(1+\eps')\opt(\sigma)$.

We finally show that when $A_v$ has to sequence a small job $J_t$, then there is a machine $M_j$ 
such that $\ell^*(j) + \ell_s(j)$ is upper bounded by $(1+\eps')\opt(\sigma)$. This implies that the
assignment of $J_t$ causes a machine load of at most $(1+\eps')\opt(\sigma) +p_t \leq (1+2\eps')\opt(\sigma) =
(1+ \eps)\opt(\sigma)$ in the final schedule $S_v$. 

So suppose that upon the arrival of a small job $J_t$ there holds $\ell^*(j) + \ell_s(j) > (1+\eps')\opt(\sigma)$
for all machines $M_j$, $1\leq j \leq m$. Recall that $\ell_s(j)$ is the load on machine $M_j$ caused by small
jobs in the current schedule $S_v$. Note that $\sum_{j=1}^m \ell^*(j)$ is the total processing time of large jobs
in $\sigma$ if processing times in $I_i$ are rounded up to $(1+\eps')^i\eps'$, for $i=1, \ldots, l$. 
Hence $1/(1+\eps)\sum_{j=1}^m \ell^*(j)$ is a lower bound on the total processing time of large jobs in $\sigma$.
It follows that the total processing time of all jobs in $\sigma$ is at least
$1/(1+\eps') \sum_{j=1}^m \ell^*(j) + \sum_{j=1}^m \ell_s(j) +p_t \geq 
1/(1+\eps') \sum_{j=1}^m (\ell^*(j) + \ell_s(j)) +p_t$. The assumption that 
$\ell^*(j) + \ell_s(j) > (1+\eps')\opt(\sigma)$ holds for all machines $M_j$ implies that
the total processing time of jobs in $\sigma$ is at least $m\cdot \opt(\sigma) + p_t > m\cdot \opt(\sigma)$, 
which contradicts the fact that $\opt(\sigma)$ is the optimum makespan.  \hspace*{\fill}{$\Box$}
\end{proof}

\section{A {$\mathbf{(4/3+\eps)}$}-competitive algorithm for \MPSO}\label{sec:4/3}

We develop an algorithm  ${\cal A}_2(\eps)$ for \MPSO\ that is $(4/3+\eps)$-competitive, for any $0<\eps \leq 1$, if
the number $m$ of machines is not too small. We then combine ${\cal A}_2(\eps)$ with ${\cal A}_1(\eps)$,
presented in the last section, and derive a strategy ${\cal A}_3(\eps)$ that is $(4/3+\eps)$-competitive,
for arbitrary $m$. The number of required schedules is $1/\eps^{O(\log (1/\eps))}$, which is a constant independent of
$n$ and $m$. We firstly present a description of the algorithm; the corresponding analysis is given thereafter.

Before describing ${\cal A}_2(\eps)$ in detail, we explain the main ideas of the algorithm. One concept
is identical to that used by ${\cal A}_1(\eps)$: Partition the range of possible job processing times into
intervals or {\em job classes\/} and consider distributions of jobs over these classes. However, in order
to achieve a constant number of schedules we have to refine this scheme and incorporate new
ideas. First, the job classes have to be chosen properly so as to allow a compact packing of jobs on the machines.
An important, new aspect in the construction of ${\cal A}_2(\eps)$ is that we will not consider 
the entire set $V$ of tuples specifying how large jobs of an input sequence~$\sigma$ are distributed over 
the job classes. Instead we will define a suitable sparsification $V'$ of $V$. Each $v\in V'$ represents 
an estimate or guess on the number of large jobs arising in $\sigma$. More specifically, if 
$v=(v_1, \ldots, v_l)$, then it is assumed that $\sigma$ contains at least $v_i$ jobs with a processing
time of job class $i$. 

Obviously, the job sequence $\sigma$ may contain more large jobs, the exact number of which is unknown. 
Furthermore, it is unknown which portion of the total processing time of $\sigma$ will arrive as small 
jobs. In order to cope with these uncertainties ${\cal A}_2(\eps)$ has to construct robust schedules. 
To this end the number of machines is partitioned into two sets ${\cal M}_c$ and
${\cal M}_r$. For the machines of ${\cal M}_c$, the algorithm initially determines a good assignment or
{\em configuration\/} assuming that $v_i$ jobs of job class $i$ will arrive. The machines of ${\cal M}_r$
are reserve machines and will be assigned additional large jobs as they arise in $\sigma$. Small jobs
will always be placed on machines in ${\cal M}_c$. The initial configuration determined for these 
machines has the property that, no matter how many small jobs arrive, a machine load never exceeds $4/3+\eps$
times the optimum makespan.

We proceed to describe ${\cal A}_2(\eps)$ in detail. Let $0<\eps \leq 1$. Moreover, set $\eps' = \eps/8$.
Again we assume without loss of generality that, for an incoming job sequence, there holds $\opt(\sigma)=1$. 
Hence the processing time of any job is upper bounded by~1. 

{\bf Job classes:} A job $J_t$, $1\leq t \leq n$, is {\em small\/}
if $p_t \leq 1/3 + 2\eps'$; otherwise $J_t$ is {\em large\/}. We divide the range of possible job processing times
into job classes. Let $I_s = (0,1/3+2\eps']$ be the interval containing the processing times of small jobs.
Let $\lambda = \lceil \log({3\over 8} + {1\over 48\eps'})\rceil$ and $l = \lambda +2$, where the logarithm is
taken to base~2. For $i=1, \ldots, l$, let 
$$\textstyle a_i = \max\{{1\over 3} -2\eps' + ({1\over 12} + {3\over 2}\eps'){1\over 2^{\lambda+1-i}}, {1\over 3} + 2\eps'\}
\ \ \ \mbox{and} \ \ \ b_i = {1\over 3} -2\eps' + ({1\over 12} + {3\over 2}\eps'){1\over 2^{\lambda-i}}.$$
It is easy to verify that $a_1 = 1/3+2\eps'$ and $a_i < b_i$, for $i=1, \ldots, l$. Furthermore
$b_{l-1} = 1/2 +\eps'$ and $b_l = 2/3 + 4\eps'$. For $i=1, \ldots, l$ define
$I_i = (a_i,b_i]$. There holds $\bigcup_{1\leq i \leq l} I_i = (1/3+2\eps', 2/3+4\eps']$. Moreover, for 
$i=1, \ldots, l-1$, let $I_{l+i} = (2a_i,2b_i]$. Intuitively, $I_{l+i}$ contains the processing times that
are twice as large as those in $I_i$, $1\leq i \leq l-1$. There holds 
$\bigcup_{1\leq i \leq l-1} I_{l+i} = (2/3+4\eps', 1+2\eps']$. Hence $I_s \cup I_1 \cup \ldots \cup I_{2l-1} =
(0,1+2\eps']$. In the following $I_i$ represents {\em job class\/} $i$, for $i=1, \ldots, 2l-1$. We say that
$J_t$ is a {\em class-$i$ job\/} if $p_t\in I_i$, where $1\leq i \leq 2l-1$.

{\bf Definition of target configurations:} As mentioned above, for any incoming job sequence $\sigma$, 
${\cal A}_2(\eps)$ works with estimates on the number of class-$i$ jobs arising in $\sigma$, $1\leq i \leq 2l-1$. 
For each estimate, the algorithm initially determines a virtual schedule or {\em target configuration\/} on a 
subset of the machines, assuming that the estimated set of large jobs will indeed arrive. Hence we partition the 
$m$ machines into two sets ${\cal M}_c$ and ${\cal M}_r$. Let $\mu = \lceil (1+\eps')/(1+2\eps')\cdot m\rceil$. Moreover, 
let ${\cal M}_c =\{M_1, \ldots, M_{\mu}\}$ and ${\cal M}_r =\{M_{\mu+1}, \ldots, M_m\}$. Set ${\cal M}_c$ 
contains the machines for which a target configuration will be computed; ${\cal M}_r$ contains the 
reserve machines. The proportion of  $|{\cal M}_r|$ to $|{\cal M}_c|$  is roughly $1:1+1/\eps'$.

A target configuration has the important property that any machine $M_j\in {\cal M}_c$ contains large jobs
of only one job class $i$, $1\leq i \leq 2l-1$. Therefore, a target configuration is properly defined
by a vector $c = (c_1, \ldots, c_{\mu}) \in \{0,\ldots,2l-1\}^{\mu}$. If $c_j=0$, then $M_j$ does not contain 
any large jobs in the target configuration, $1\leq j\leq \mu$. If  $c_j=i$, where 
$i\in \{1,\ldots, 2l-1\}$, then $M_j$ contains class-$i$ jobs, $1\leq j\leq \mu$. The vector $c$ 
implicitly also specifies how many large jobs reside on a machine. If $c_j=i$ with $1\leq i \leq l$, then
$M_j$ contains two class-$i$ jobs. Note that, for general $i\in \{1,\ldots, l\}$, a third job cannot be 
placed on the machine without exceeding a load bound of $4/3+\eps$. If $c_j=i$ with $l+1\leq i \leq 2l-1$, then
$M_j$ contains one class-$i$ job. Again, the assignment of a second job is not feasible in general. 
Given a configuration $c$, $M_j$ is referred to as a {\em class-$i$ machine} if $c_j=i$, where
$1\leq j \leq \mu$ and $1\leq i \leq 2l-1$. 

With the above interpretation of target configurations, each vector  $c = (c_1, \ldots, c_{\mu})$ encodes
inputs containing $2|\{c_j \in\{c_1,\ldots c_{\mu}\} : c_j=i\}|$ class-$i$ jobs, for $i=1,\ldots, l$, as
well as  $|\{c_j\in \{c_1,\ldots c_{\mu}\} : c_j=i\}|$ class-$i$ jobs, for $i=l+1,\ldots, 2l-1$.  Hence,
for an incoming job sequence, instead of considering estimates on the number of class-$i$ jobs, for any
$1\leq i \leq 2l-1$, we can equivalently consider target configurations. Unfortunately, it
will not be possible to work with all target configurations $c\in \{0,\ldots,2l-1\}^{\mu}$ since the
resulting number of schedules to be constructed would be $(2l)^{\mu} = (\log(1/\eps))^{\Omega(m)}$.
Therefore, we will work with a suitable sparsification of the set of all configurations.

{\bf Sparsification of the set of target configurations:} Let $\kappa = \lceil 2(2+1/\eps')(2l-1)\rceil$ and
$U = \{0.\ldots, \kappa\}^{2l-1}$. We will show that 
$\kappa \lfloor (m-\mu)/(2l-1) \rfloor \geq m$ if $m$ is not too small (see Lemma~\ref{lem:kappa}).
This property in turn will ensure that any job sequence $\sigma$ can be mapped to a $u\in U$.
For any vector $u = (u_1, \ldots, u_{2l-1})\in U$, we define a target configuration
$c(u)$ that contains $u_i \lfloor (m-\mu)/(2l-1) \rfloor$ class-$i$ machines, for $i=1, \ldots, 2l-1$, 
provided that $\sum_{i=1}^{2l-1} u_i \lfloor (m-\mu)/(2l-1) \rfloor$ does not exceed $\mu$. More specifically,
for any $u = (u_1, \ldots, u_{2l-1})\in U$, let $\pi_0=0$ and $\pi_i = \sum_{j=1}^i u_j \lfloor (m-\mu)/(2l-1) \rfloor$,
be the partial sums of the first $i$ entries of $u$, multiplied by $\lfloor (m-\mu)/(2l-1) \rfloor$, for
$i=1, \ldots, 2l-1$. Let $\mu'= \pi_{2l-1}$. First construct a vector $c'(u) = (c'_1, \ldots, c'_{\mu'})$
of length $\mu'$ that contains exactly $u_i \lfloor (m-\mu)/(2l-1) \rfloor$ class-$i$ machines.
That is, for $i=1, \ldots, 2l-1$, let $c'_j = i$ for $j=\pi_{i-1}+1, \ldots, \pi_i$. We now truncate
or extend $c'(u)$ to obtain a vector of length $\mu$. If $\mu'\geq \mu$, then $c(u)$ is the vector consisting
of the first $\mu$ entries of $c'(u)$. If $\mu'< \mu$, then $c(u) = (c'_1, \ldots, c'_{\mu'}, 0, \ldots, 0)$, i.\,e.\
the last $\mu-\mu'$ entries are set to~0. Let $C= \{c(u) \mid u\in U\}$ be the set of all target configurations
constructed from vectors $u\in U$.

{\bf The algorithm family:} Let ${\cal A}_2(\eps) = \{A_c\}_{c\in C}$. For any $c\in C$, algorithm
$A_c$ works as follows. Initially, prior to the arrival of any job of $\sigma$, $A_c$ determines
the target configuration specified by $c= (c_1, \ldots, c_{\mu})$ and uses this virtual schedule for
the machines of ${\cal M}_c$ to make scheduling decisions. Consider a machine $M_j\in {\cal M}_c$ and
suppose $c_j >0$, i.\,e.\ $M_j$ is a class-$i$ machine  for some $i\geq 1$. Let $\ell^-(j)$ and $\ell^+(j)$
be the targeted minimal and maximal loads caused by large jobs on $M_j$, according to the target
configuration. More precisely, if $i\in \{1,\ldots, l\}$, then  $\ell^-(j) = 2a_i$ and  $\ell^+(j)=2b_i$. 
Recall that in a target configuration a class-$i$ machine contains two class-$i$ jobs if $1\leq i \leq l$.
If $i\in \{l+1,\ldots, 2l-1\}$ and hence $i = l+i'$ for some $i'\in \{1,\ldots, l-1\}$, then  $\ell^-(j) = 2a_{i'}$ 
and  $\ell^+(j)=2b_{i'}$. If $M_j \in {\cal M}_c$ is 
a machine with $c_j = 0$, then $ \ell^-(j)  = \ell^+(j) = 0$. While the job sequence $\sigma$ is processed,
a machine $M_j\in {\cal M}_c$ may or may not be {\em admissible\/}. Again assume that $M_j$ is a class-$i$
machine with $i\geq 1$. If $i\in \{1, \ldots, l\}$, then at any time during the scheduling process $M_j$ is
admissible if it has received less than two class-$i$ jobs so far. Analogously, if  $i\in \{l+1, \ldots, 2l-1\}$,
then $M_j$ is admissible if it has received no class-$i$ job so far. Finally, at any time during the
scheduling process, let $\ell(j)$ be the current load of machine $M_j$ and let $\ell_s(j)$ be the
load due to small jobs, $1\leq j \leq m$.

Algorithm $A_c$ schedules each incoming job $J_t$, $1\leq t \leq n$, in the following way. First assume
that $J_t$ is a large job and, in particular, a class-$i$ job, $1\leq i \leq 2l-1$. The algorithm checks
if there is a class-$i$ machine in ${\cal M}_c$ that is admissible. If so, $J_t$ is assigned to such a machine.
If there is no admissible class-$i$ machine available, then $J_t$ is placed on a machine in ${\cal M}_r$.
There jobs are scheduled according to the {\em Best-Fit\/} policy. More specifically, $A_c$ checks if there exists
a machine $M_j\in {\cal M}_r$ such that $\ell(j) +p_t \leq 4/3+\eps$. If this is the case, then $J_t$ is
assigned to such a machine with the largest current load $\ell(j)$. If no such machine exists, $J_t$
is assigned to an arbitrary machine in ${\cal M}_r$. Next assume that $J_t$ is small. The job is a assigned
to a machine in ${\cal M}_c$, where preference is given to machines that have already received small jobs. 
Algorithm $A_c$ checks if there is an $M_j\in {\cal M}_c$ with $\ell_s(j) >0$ such that 
$\ell^+(j) + \ell_s(j) +p_t \leq 4/3+\eps$. If this is the case, then $J_t$ is assigned to any such machine.
Otherwise $A_c$ considers the machines of ${\cal M}_c$ which have not yet received any small jobs. If there exists
an $M_j\in {\cal M}_c$ with $\ell_s(j) = 0$ such that $\ell^+(j) +p_t \leq 4/3+\eps$, then among these
machines $J_t$ is assigned to one having the smallest targeted load $\ell^-(j)$. 
If again no such machine exists,  $J_t$ is assigned to an arbitrary machine in ${\cal M}_c$. A summary of  
${\cal A}_2(\eps)$, which focuses on the job assignment rules, is given in Figure~\ref{fig:3}. We obtain the following result.

\begin{figure}[h]
\fbox{
\begin{minipage}{11.7cm}
{\bf Algorithm ${\cal A}_2(\eps)$} \\[-15pt]
\begin{itemize}
\item[1.] ${\cal A}_2(\eps) = \{A_c\}_{c\in C}$, where $C=\{c(u) \mid u\in U \}$\\
 $U = \{0,\ldots, \kappa\}^{2l-1}$, where $\kappa = \lceil 2(2+1/\eps')(2l-1)\rceil$,  
 $l = \lceil \log({3\over 8} + {1\over 48\eps'})\rceil +2$ and $\eps'= \eps/8$\\
 $\mu = \lceil (1+\eps')/(1+2\eps')\cdot m\rceil$\\[-8pt]
\item[2.] $A_c$ works as follows.\\[-8pt]
\begin{itemize}
\item[(a)] Determine target configuration specified by $c = (c_1, \ldots, c_{\mu})$.\\[-8pt]
\item[(b)] Each $J_t$ is sequenced as follows.\\
{\em $J_t$ is large:\/} Let $J_t$ be a class-$i$ job. If there is an admissible class-$i$ machine
in ${\cal M}_c$, assign $J_t$ to it. Otherwise check if $\exists$ $M_j \in {\cal M}_r$ such that 
$\ell(j) + p_t \leq 4/3+\eps$. If so, assign $J_t$ to such an $M_j$ with the highest $\ell(j)$; otherwise place 
$J_t$ on an arbitrary $M_j\in {\cal M}_r$.\\
{\em $J_t$ is small:\/} If $\exists$ $M_j\in {\cal M}_c$ with $\ell_s(j) >0$ such that  $\ell^+(j) + \ell_s(j)+p_t
\leq 4/3+\eps$, assign $J_t$ to it. Otherwise check if $\exists$ $M_j \in {\cal M}_c$ with $\ell_s(j) =0$ such 
that  $\ell^+(j) + p_t \leq 4/3+\eps$. If so, assign $J_t$ to such an $M_j$ with the lowest $\ell^-(j)$; otherwise place 
$J_t$ on an arbitrary $M_j\in {\cal M}_c$.\\[-8pt]
\end{itemize}
\end{itemize}
\end{minipage} 
}
\caption{The algorithm ${\cal A}_2(\eps)$}\label{fig:3}
\end{figure}

\begin{theorem}\label{th:alg2}
${\cal A}_2(\eps)$ is $(4/3+\eps)$-competitive, for any $0< \eps \leq 1$ and $m \geq 2l/(\eps')^2$. The algorithm
uses $1/\eps^{O(\log (1/\eps))}$ schedules. 
\end{theorem}

${\cal A}_2(\eps)$ is $(4/3+\eps)$-competitive if, for the chosen $\eps$, the number of machines is at least
$2l/(\eps')^2$. If the number of machines is smaller, we can simply apply algorithm ${\cal A}_1(\eps)$
with an accuracy of $\eps_0 = 1/3$. Let ${\cal A}_3(\eps)$ be the following combined algorithm. If for the
chosen $\eps$, $m< 2l/(\eps')^2$, execute ${\cal A}_1(1/3)$. Otherwise execute ${\cal A}_2(\eps)$.

\begin{corollary}\label{cor:A3}
${\cal A}_3(\eps)$ is $(4/3+\eps)$-competitive, for any $0< \eps \leq 1$, and uses $1/\eps^{O(\log (1/\eps))}$ 
schedules. 
\end{corollary}
\begin{proof}
If ${\cal A}_1(1/3)$ is executed for a machine number $m< 2l/(\eps')^2$, then by Theorem~\ref{th:guess2}
the number of schedules is $(\log (1/\eps)/\eps^3)^{O(1)}$, which is $1/\eps^{O(1)}$. \hspace*{\fill}{$\Box$}
\end{proof}

In the remainder of this section we prove Theorem~\ref{th:alg2}. The stated number of schedules follows from the 
fact that ${\cal A}_2(\eps)$  consists of $|C| = (\kappa+1)^{2l-1}$ algorithms. Recall that 
$\kappa = \lceil 2(2+1/\eps')(2l-1)\rceil$ and
$l = \lambda +2 = \lceil \log({3\over 8} + {1\over 48\eps'})\rceil +2$. Hence $l = O(\log (1/\eps))$
and $\kappa = O(1/\eps \log (1/\eps))$, which gives that $|C|$ is $1/\eps^{O(\log (1/\eps))}$.

Hence it suffices to show that, for any job sequence $\sigma$, ${\cal A}_2(\eps)$ generates
a schedule whose makespan is at most $(4/3+\eps)\opt(\sigma)$, which we will do in the remainder of this section. More specifically we will prove that, for any $\sigma$,
there exists a target configuration $c\in C$ that accurately models the large jobs arising in $\sigma$. We
will refer to such a vector as a valid target configuration. Then we will show that the corresponding algorithm 
$A_c$ builds a schedule with a makespan of at most $(4/3+\eps)\opt(\sigma)$. 

We introduce some notation. Consider any job sequence $\sigma$. For any $i$, $1\leq i \leq 2l-1$, let
$n_i(\sigma)$ be the number of class-$i$ jobs arising in $\sigma$, i.\,e.\ $n_i(\sigma)$ is the number of
jobs $J_t$ with $p_t\in I_i$. Furthermore, for any target configuration $c=(c_1, \ldots, c_{\mu})\in C$ and 
any $i$ with $1\leq i \leq 2l-1$, let $m_i$ be the number of class-$i$ machines in $c$, i.\,e.\ 
$m_i = |\{c_j\in \{c_1, \ldots, c_{\mu}\} : c_j = i\}|$. Let $\mu_1 = \sum_{i=1}^l m_i$ be the total 
number of class-$i$ machines with $i\in \{1, \ldots, l\}$. Similarly,  $\mu_2 = \sum_{i=l+1}^{2l-1} m_i$
is the total number of class-$i$ machines with $i\in \{l+1, \ldots, 2l-1\}$. Given $\sigma$, vector
$c\in C$ will be a valid target configuration if, for any $i=1, \ldots, 2l-1$, $\sigma$ contains as
many class-$i$ jobs as specified in $c$ and, moreover, if all the additional large jobs can be feasibly
scheduled on the $m-\mu$ reserve machines. Recall that in a configuration $c$, any class-$i$ machine
with $1\leq i \leq l$ is supposed to contain two class-$i$ jobs. Formally, $c\in C$ is a {\em valid
target configuration\/} if the following three conditions hold.
\begin{enumerate}[(i)]
 \item For $i=1,\ldots, l$, there holds $2m_i \leq n_i(\sigma)$.
 \item For $i=l+1,\ldots, 2l-1$, there holds $m_i \leq n_i(\sigma)$.
 \item $\lceil (\sum_{i=1}^l n_i(\sigma) -2\mu_1) / 2 \rceil + \sum_{i=l+1}^{2l-1} n_i(\sigma) - \mu_2 \leq
        m-\mu$
\end{enumerate}
Conditions (i) and (ii) represent the constraint that $\sigma$ contains as many class-$i$ jobs as specified
in $c$, $1\leq i \leq 2l-1$. Condition (iii) models the requirement that extra large jobs can be feasibly packed
on the reserve machines. Here $\sum_{i=1}^l n_i(\sigma) -2\mu_1$ is the extra number of class-$i$ jobs with
$i\in \{1, \ldots, l\}$ in $\sigma$. Any two of these can be packed on one machine since the processing time
of any of these jobs is upper bounded by $b_l \leq 2/3+4\eps'$. Hence two jobs incur a machine load of
at most $4/3+8\eps' = 4/3 +\eps$. Analogously, $\sum_{i=l+1}^{2l-1} n_i(\sigma) - \mu_2$ is the extra 
number of class-$i$ jobs with  $i\in \{l+1, \ldots, 2l-1\}$, which cannot be combined together because
their processing times are greater than $2a_1 \geq 2/3 + 4\eps'$. 

In order to prove that, for any $\sigma$, there exists a valid target configuration we need two lemmas.
\begin{lemma}\label{lem:jobs}
For any $\sigma$, there holds $\lceil \sum_{i=1}^l n_i(\sigma) / 2 \rceil + \sum_{i=l+1}^{2l-1} n_i(\sigma) \leq m$. 
\end{lemma}
\begin{proof}
Consider any optimal schedule $S^*$ for $\sigma$ and recall that we assume without loss of generality that
$\opt(\sigma)=1$. In $S^*$ any machine containing a class-$i$ job with $i\in \{l+1, \ldots, 2l-1\}$ cannot
contain an additional large job: The class-$i$ job causes a load greater than $2a_1 \geq 2/3+4\eps'$ and
any additional large job, having a processing time greater than $1/3+2\eps'$, would generate a total load
greater than~1. Furthermore, any machine containing a class-$i$ job with $i\in \{1,\ldots, l\}$ can
contain at most one additional job of the job classes $1, \ldots, l$ because two further jobs 
would generate a total load greater than $3a_1 \geq 3(1/3+2\eps') >1$. \hspace*{\fill}{$\Box$}
\end{proof}

\begin{lemma}\label{lem:kappa}
For any $0<\eps'\leq 1/8$, there holds $\kappa \lfloor (m-\mu)/(2l-1)\rfloor \geq m$ if $m\geq 2l/(\eps')^2$. 
\end{lemma}
\begin{proof}
There holds
\begin{eqnarray*}
\kappa \lfloor (m-\mu)/(2l-1)\rfloor &\geq 
 & 2 \textstyle{(2+\frac{1}{\eps'})} (2l-1)   \cdot    \lfloor    (m-\mu)/(2l-1) \rfloor   \\
 & \geq & 2 \textstyle{(2+\frac{1}{\eps'})}  (2l-1)  \cdot   (  (m-\mu)/(2l-1)  -1  )   \notag \\
& = & 2 \textstyle{(2+\frac{1}{\eps'})}  (2l-1)   \cdot \left (  \dfrac {m-\lceil \frac{1+\eps'}{1+2\eps'}m  \rceil }{ 2l-1 } -1  \right)     \notag  \\
&  \geq & 2 \textstyle{(2+\frac{1}{\eps'})}  (2l-1) \cdot \left (  \dfrac {  \frac{\eps'}{1+2\eps'}m  -1 }{ 2l-1 } -1  \right)    \notag  \\ 
& = & 2 \textstyle{(2+\frac{1}{\eps'})} (2l-1)  \cdot \left ( \frac { \eps'm  - (1+2\eps')2l }{(1+2\eps')(2l-1)} \right) \notag \\
&\geq & m + m - (2/\eps')(1+2\eps') 2l \notag \\
& \geq & m, \notag
\end{eqnarray*}
where the last line follows because of $m \geq 2l/(\eps')^2$ and $2l/(\eps')^2 \geq (2/\eps')(1+ 2\eps') 2l$, for
any $\eps'\leq 1/8$. \hspace*{\fill}{$\Box$}
\end{proof}

The next lemma establishes the existence of valid target configurations.
\begin{lemma}\label{lem:config}
For any $\sigma$, there exists a valid target configuration $c\in C$ if $m \geq 2l/(\eps')^2$.  
\end{lemma}
\begin{proof}
In this proof let $m_0 = \lfloor (m-\mu)/(2l-1)\rfloor$. Given $\sigma$, we first construct a vector 
$u\in U$. Lemma~\ref{lem:jobs} implies that for any job class $i$, $1\leq i \leq l$, there holds
$\lceil n_i(\sigma)/2\rceil \leq m$. For any job class $i$, $l+1\leq i \leq 2l-1$, there holds $n_i(\sigma) \leq m$.
By Lemma~\ref{lem:kappa}, $\kappa m_0\geq m$, which is equivalent to $m/m_0 \leq \kappa$. For any $i$
with $1\leq i \leq l$, set $u_i = \lfloor n_i(\sigma)/(2m_0)\rfloor$. For any $i$ with 
$l+1\leq i \leq 2l-1$, set $u_i = \lfloor n_i(\sigma)/m_0\rfloor$. Then $u_i \in \{0,\ldots, \kappa\}$, for
$i=1,\ldots, 2l-1$, and the resulting vector $u=(u_1, \ldots. u_{2l-1})$ is element of $U$. We next show
that the vector $c(u)$ constructed by ${\cal A}_2(\eps)$ is a valid target configuration.

When ${\cal A}_2(\eps)$ constructs $c(u)$, it first builds a vector $c'(u)= (c'_1,\ldots, c'_{\mu'})$ of
length $\mu' = \sum_{i=1}^{2l-1} u_im_0$ containing exactly $u_im_0$ entries with $c'_j = i$, for
$i=1,\ldots, 2l-1$. If $\mu'\geq \mu$, then $c(u)$ contains the first $\mu$ entries of $c'(u)$. If 
$\mu'<\mu$, then $c(u)$ is obtained from $c'(u)$ by adding $\mu-\mu'$ entries of value~0. In either case 
$c(u)$ contains at most $u_im_0$ entries of values $i$, for $i=1,\ldots, 2l-1$. Hence for the target
configuration $c(u)$, there holds $m_i\leq u_im_0$, for $i=1,\ldots, 2l-1$, where $m_i$ is again the total
number of class-$i$ machines in $c(u)$. 

If $i\in\{1,\ldots, l\}$, then $m_i \leq \lfloor n_i(\sigma)/(2m_0)\rfloor m_0 \leq n_i(\sigma)/2$,
which is equivalent to $2m_i \leq n_i(\sigma)$. Similarly, if $i\in \{l+1, \ldots, 2l-1\}$, then 
$m_i \leq \lfloor n_i(\sigma)/m_0\rfloor m_0 \leq n_i(\sigma)$. Therefore, conditions~(i) and (ii) 
defining valid target configurations are satisfied and we are left to verify
condition~(iii).

First assume $\mu'\geq \mu$. In this case the vector $c(u)$ contains no entries of value~0 and hence 
$\mu = \mu_1+\mu_2$. Recall that $\mu_1 = \sum_{i=1}^l m_i$ is the total number of class-$i$ machines 
with $i\in \{1,\ldots, l\}$ specified in $c(u)$. Similarly, $\mu_2 = \sum_{i=l+1}^{2l-1} m_i$ is
the total number of class-$i$ machines with $i\in \{l+1,\ldots, 2l-1\}$. By Lemma~\ref{lem:jobs}, 
$\lceil \sum_{i=1}^l n_i(\sigma) / 2 \rceil + \sum_{i=l+1}^{2l-1} n_i(\sigma) \leq m$. Subtracting
the equation $\mu_1+\mu_2 = \mu$, we obtain 
$$\textstyle \lceil \sum_{i=1}^l n_i(\sigma) / 2 \rceil -\mu_1 + \sum_{i=l+1}^{2l-1} n_i(\sigma) -\mu_2 
\leq m-\mu.$$
There holds $\lceil \sum_{i=1}^l n_i(\sigma) / 2 \rceil -\mu_1 = \lceil (\sum_{i=1}^l n_i(\sigma) -2\mu_1)/ 2 \rceil$
because $\mu_1$ is an integer. Hence condition~(iii) defining valid target configurations is 
satisfied.

It remains to study the case $\mu' < \mu$. For any $i$ with $i\in \{l+1,\ldots, 2l-1\}$, there holds
$u_i = \lfloor n_i(\sigma)/m_0 \rfloor$ and hence $u_i > n_i(\sigma)/m_0 -1$, which is equivalent to
$n_i(\sigma) < (u_i+1)m_0$. Hence
$$\textstyle \sum_{i=l+1}^{2l-1} n_i(\sigma) < \sum_{i=l+1}^{2l-1} (u_i+1)m_0 = \sum_{i=l+1}^{2l-1} u_im_0 + (l-1)m_0.$$
The sum $\sum_{i=l+1}^{2l-1} u_im_0 = \sum_{i=l+1}^{2l-1} u_i \lfloor (m-\mu)/(2l-1)\rfloor$ is the total
number of entries $c'_j$ with $c'_j\in \{l+1,\ldots, 2l-1\}$ in $c'(u)$. Since $\mu' <\mu$, none of these
entries is deleted when $c(u)$ is derived from $c'(u)$. Hence $\sum_{i=l+1}^{2l-1} u_im_0 =\mu_2$ 
is the total number of class-$i$ machines with $i\in \{l+1,\ldots, 2l-1\}$ specified in $c(u)$. We conclude
\begin{equation}\label{eq:n1}
\textstyle \sum_{i=l+1}^{2l-1} n_i(\sigma) \leq \mu_2 + (l-1)m_0.
\end{equation}

For any $i$ with $i\in \{1,\ldots, l\}$, there holds $u_i = \lfloor n_i(\sigma)/(2m_0) \rfloor$ and hence
$u_i > n_i(\sigma)/(2m_0) -1$. This implies $n_i(\sigma)/2 < (u_i+1)m_0$. Since $(u_i+1)m_0$ is an integer we
obtain $n_i(\sigma)/2 \leq  (u_i+1)m_0 -1$. Thus 
\begin{equation}\label{eq:n2}
\textstyle \lceil \sum_{i=1}^l n_i(\sigma)/2\rceil \leq \sum_{i=1}^l n_i(\sigma)/2 +1 \leq 
\sum_{i=1}^l (u_i+1)m_0 = \mu_1 +lm_0.
\end{equation}
Again $\sum_{i=1}^l u_im_0 = \mu_1$ because $c'(u)$ contains exactly  $\sum_{i=1}^l u_im_0$ entries
$c'_j$ with $c'_j\in \{1,\ldots, l\}$ and all of these entries are contained in $c(u)$ representing class-$i$ 
machines for $i\in \{1,\ldots, l\}$. 
Inequalities~(\ref{eq:n1}) and (\ref{eq:n2}) together with the identity $m_0 = \lfloor (m-\mu)/(2l-1)\rfloor$
imply
$$\textstyle \lceil \sum_{i=1}^l n_i(\sigma)/2\rceil -\mu_1 + \sum_{i=l+1}^{2l-1} n_i(\sigma) - \mu_2 \leq
(2l-1) \lfloor (m-\mu)/(2l-1)\rfloor \leq m-\mu.$$
Since again $\lceil \sum_{i=1}^l n_i(\sigma)/2\rceil -\mu_1 = \lceil (\sum_{i=1}^l n_i(\sigma)-2\mu_1)/2\rceil$,
condition~(iii) defining valid target configurations holds. \hspace*{\fill}{$\Box$}
\end{proof}

We next analyze the scheduling steps of ${\cal A}_2(\eps)$. 
\begin{lemma}\label{lem:sched1}
Let $A_c$ be any algorithm of ${\cal A}_2(\eps)$ processing a job sequence $\sigma$. At any time there
exists at most one machine $M_j\in {\cal M}_c$ with $\ell_s(j) >0$ and $\ell^-(j) + \ell_s(j) < 1+\eps'$
in the schedule maintained by $A_c$.
\end{lemma}
\begin{proof}
Consider any point in time while $A_c$ sequences $\sigma$. Suppose that there exists a machine 
$M_j\in {\cal M}_c$ with $\ell_s(j) >0$ and $\ell^-(j) + \ell_s(j) < 1+\eps'$. We show that if a small
job $J_t$ arrives and $A_c$ assigns it to a machine $M_{j'} \in {\cal M}_c$ with $\ell_s(j') =0$,
then $\ell^-(j')+p_t >1+\eps'$ so that no new machine with the property specified in the lemma is
generated. A first observation is that $M_j$ is not a class-$l$ machine because in this case $\ell^-(j)$ 
would be $2a_l = 2b_{l-1} = 1 +2\eps'$. Also, if $M_{j'}$ is a class-$l$ machine, there is 
nothing to show because, again, in this case $\ell^-(j') \geq  1+2\eps'$.

So assume that $A_c$ assigns $J_t$ to a machine $M_{j'} \in {\cal M}_c$, which is not a class-$l$ machine,
and $\ell_s(j') =0$ prior to the assignment. We first show that $\ell^-(j') \geq \ell^-(j)$. Consider
the scheduling step in which $A_c$ assigned the first small job $J_{t'}$ to $M_j$. Since $M_j$ is not
a class-$l$ machine $\ell^+(j) = 2b_i$, for some $i\in \{1,\ldots, l-1\}$ and the assignment of
$J_{t'}$ to $M_j$ led to a load of at most $\ell^+(j) + p_{t'} \leq 1 + 2\eps' + 1/3+2\eps' 
= 4/3+4\eps' < 4/3 +\eps$. Since $M_{j'}$ is not a class-$l$ machine either, $J_{t'}$ could have also
been assigned to $M_{j'}$ incurring a resulting load of at most $\ell^+(j')+p_{t'} < 4/3 +\eps$ on this
machine. Note that when an algorithm $A_c$ cannot assign a small job to a machine $M_j\in {\cal M}_c$ 
with $\ell_s(j) >0$ and instead has to resort to machines $M_k\in {\cal M}_c$  with $\ell_s(k) =0$, it
chooses a machine having the smallest $\ell^-(k)$ value. We conclude $\ell^-(j) \leq \ell^-(j')$.

Next consider the assignment of $J_t$. Algorithm $A_c$ would prefer to place $J_t$ on $M_j$ as it already
contains small jobs. Since this is impossible, there holds $\ell^+(j) + \ell_s(j) +p_t > 4/3+\eps$ and thus
$p_t > 4/3+ 8\eps' - \ell^+(j) - \ell_s(j)$. Since by assumption $\ell^-(j) + \ell_s(j) <1+\eps'$ it follows
$p_t > 1/3 + 7\eps' - \ell^+(j) + \ell^-(j)$. Suppose that $\ell^+(j) = 2b_i$, for some 
$i\in \{1,\ldots, l-1\}$. Then $\ell^-(j) = 2a_i$. Since $\ell^-(j') \geq \ell^-(j)$ we obtain
\begin{eqnarray*}
\ell^-(j') +p_t &\geq& 1/3 + 7\eps' + \ell^-(j) - \ell^+(j) + \ell^-(j)\\
& \geq & {\textstyle 1/3 + 7\eps' + 2({1\over 12} + {3\over 2}\eps')({1\over 2^{\lambda+1-i}} - {1\over 2^{\lambda-i}})}\\
& & {\textstyle +2/3 - 4\eps' +  2({1\over 12} + {3\over 2}\eps'){1\over 2^{\lambda+1-i}}}\\
& = & 1 + 3\eps' > 1+\eps',
\end{eqnarray*}
as desired. \hspace*{\fill}{$\Box$}
\end{proof}

The following lemmas focus on algorithms $A_c$ such that $c$ is a valid target configuration for $\sigma$.
\begin{lemma}\label{lem:sched2}
Let $\sigma$ be any job sequence and $A_c$ be an algorithm such that $c$ is a valid target configuration for
$\sigma$. Let $m \geq 2l/(\eps')^2$. Consider any point in time during the scheduling process. If the
schedule of $A_c$ contains at most one machine $M_j\in {\cal M}_c$ with $\ell^-(j)+\ell_s(j) < 1+\eps'$, then
no further small job can arrive. 
\end{lemma}
\begin{proof}
Since $c$ is a valid target configuration for $\sigma$, the job sequence contains as many class-$i$ jobs,
for any $i\in \{1,\ldots, l\}$, as indicated by $c$. Hence the total processing time of large jobs in 
$\sigma$ is lower bounded by $\sum_{j=1}^{\mu} \ell^-(j)$. Hence the total processing time of jobs in
$\sigma$ is at least $\sum_{j=1}^{\mu} (\ell^-(j) + \ell_s(j))$, where the machine loads due to small
jobs may be considered at an arbitrary point in time. Hence if there exists a time such that 
$\ell_s(j) + \ell^-(j) < 1+\eps'$ for at most one $M_j \in {\cal M}_c$, we obtain
\begin{eqnarray*}
{\textstyle \sum_{j=1}^{\mu} (\ell^-(j) + \ell_s(j))} & \geq & 
\textstyle{(1+\eps')(\mu-1) \geq (1+\eps') ({1+\eps' \over 1+2\eps'} m -1)}\\
& = & {\textstyle m + {(\eps')^2\over1+2\eps'} m - (1+\eps') \geq m.}
\end{eqnarray*}
The last inequality holds because $m \geq 2l/(\eps')^2 \geq 2/(\eps')^2 \geq (1+\eps')(2\eps'+1)/(\eps')^2$,
for any $\eps' \leq 1/8$. Hence no further small job can arrive. \hspace*{\fill}{$\Box$}
\end{proof}

\begin{lemma}\label{lem:sched3}
Let $\sigma$ be any job sequence and $A_c$ be an algorithm such that $c$ is a valid target configuration for
$\sigma$. Let $m \geq 2l/(\eps')^2$. Then in the final schedule constructed by $A_c$, each machine in 
${\cal M}_c$ has a load of at most $4/3 +\eps$. 
\end{lemma}
\begin{proof}
We consider the scheduling steps in which $A_c$ assigns a job $J_t$ to a machine in ${\cal M}_c$. First suppose 
that $J_t$ is large. Let $J_t$ be a class-$i$ job, where $1\leq i \leq 2l-1$. If $J_t$ is assigned to an 
$M_j \in {\cal M}_c$, then $M_j$ must be an admissible class-$i$ machine, i.\,e.\ prior to the assignment of $J_t$ 
it contains fewer class-$i$ jobs as specified by the target configuration. This implies that for any machine
$M_j \in {\cal M}_c$, its load due to large jobs is always at most $\ell^+(j)$. The latter value is upper
bounded by $2b_l \leq 2(2/3+4\eps') = 4/3 + 8\eps' = 4/3+\eps$. Hence, in order to establish the lemma it
suffices to show that whenever a small job is assigned to a machine $M_j \in {\cal M}_c$, the resulting
load $\ell^+(j) + \ell_s(j)$ on $M_j$ is at most $4/3+\eps$.

Suppose on the contrary that a small job $J_t$ arrives and $A_c$ schedules it on a machine in ${\cal M}_c$
such that the resulting load is greater than $4/3+\eps$. Algorithm $A_c$ first tries to place $J_t$ on a 
machine $M_j \in {\cal M}_c$ with $\ell_s(j) >0$, which has already received small jobs. By Lemma~\ref{lem:sched1},
among these machines there exists at most one having the property that $\ell^-(j) + \ell_s(j) < 1+\eps'$. 
Since an assignment to those machines is impossible without exceeding a load of $4/3+\eps$, $A_c$
tries to place $J_t$ on a machine $M_j\in {\cal M}_c$ with $\ell_s(j) =0$. Since this is also impossible
without exceeding a load of $4/3+\eps$, any $M_j\in {\cal M}_c$ with $\ell_s(j)=0$ must be a class-$l$
machine. This holds true because for any class-$i$ machine with $i\neq l$, there holds $\ell^+(j) \leq 2b_{l-1}
\leq 1+2\eps'$ and an assignment of a small job would result in a total load of at most 
$1+2\eps' + 1/3 + 2\eps' < 4/3+\eps$. Observe that any class-$l$ machine has a targeted minimal load of
$2a_l = 2b_{l-1} \geq 1+2\eps' > 1+\eps'$. 

We conclude that immediately before the assignment of $J_t$ the schedule of $A_c$ contains at most one
machine $M_j \in {\cal M}_c$ with $\ell^-(j) + \ell_s(j) < 1+\eps'$. Lemma~\ref{lem:sched2} implies 
that the incoming job $J_t$ cannot be small, and we obtain a contradiction. \hspace*{\fill}{$\Box$}
\end{proof}

\begin{lemma}\label{lem:sched4}
Let $\sigma$ be any job sequence and $A_c$ be an algorithm such that $c$ is a valid target configuration for
$\sigma$. Then in the final schedule constructed by $A_c$, each machine in ${\cal M}_r$ has a load of at
most $4/3 +\eps$. 
\end{lemma}
\begin{proof}
Algorithm $A_c$ assigns only large jobs to machines in ${\cal M}_r$. A first observation is that whenever there
exists an $M_j \in {\cal M}_r$ that contains only one class-$i$ job with $i\in \{1,\ldots, l\}$ but no further
jobs, then an incoming class-$i'$ job with $i'\in \{1,\ldots, l\}$ will not be assigned to an empty machine.
This holds true because the two jobs can be combined, which results in a total load of at most $2b_l \leq 4/3+8\eps'
=4/3+\eps$.

The observation implies that at any time while $A_c$ processes $\sigma$, the number of machines of ${\cal M}_r$
containing at least one job is upper bounded by $\lceil n_1/2\rceil +n_2$. Here $n_1$ denotes the total
number of class-$i$ jobs with $i\in \{1,\ldots, l\}$ that have been assigned to machines of ${\cal M}_r$ so far.
Analogously, $n_2$ is the total number of class-$i$ jobs with $i\in \{l+1,\ldots, 2l-1\}$ currently residing
on machines in ${\cal M}_r$. Since $c$ is a valid target configuration for $\sigma$ conditions~(i) and (ii)
defining those configurations imply $0\leq \sum_{i=1}^l n_i(\sigma) -2\mu_1$ and 
$0\leq \sum_{i=l+1}^{2l-1} n_i(\sigma) -\mu_2$. Moreover, since $A_c$ assigns large jobs preferably to machines 
in ${\cal M}_c$, there holds $n_1 \leq  \sum_{i=1}^l n_i(\sigma) -2\mu_1$ and  
$n_2 \leq \sum_{i=l+1}^{2l-1} n_i(\sigma) -\mu_2$. By condition~(iii) defining valid target configurations,
$\lceil (\sum_{i=1}^l n_i(\sigma) -2\mu_1) / 2 \rceil + \sum_{i=l+1}^{2l-1} n_i(\sigma) - \mu_2 \leq m-\mu$.
Hence, while $n_2 < \sum_{i=l+1}^{2l-1} n_i(\sigma) -\mu_2$ there holds $\lceil n_1/2\rceil +n_2 < m-\mu$
and thus exists an empty machine ${\cal M}_r$
to which an incoming class-$i$ jobs with $i\in \{l+1,\ldots, 2l-1\}$ can be assigned. Similarly, while 
$n_1 < \sum_{i=1}^l n_i(\sigma) -2\mu_1$, there must exist an empty machine or a machine containing only
one class-$i'$ job with $i'\in \{1,\ldots, l\}$ to which in incoming class-$i$ job with $i\in \{1,\ldots, l\}$
can be assigned. In either case, the assignment generates a load of at most $4/3+\eps$ on the selected 
machine. \hspace*{\fill}{$\Box$}
\end{proof}
Theorem~\ref{th:alg2} now follows from Lemmas~\ref{lem:config}, \ref{lem:sched3} and \ref{lem:sched4}.

\section{Algorithms for \MPS}\label{sec:mps}

We derive our algorithms for \MPS. The strategies are obtained by simply combining ${\cal A^*}(\eps)$,
presented in Section~\ref{sec:redu}, with ${\cal A}_1(\eps)$ and ${\cal A}_3(\eps)$. 
In order to achieve a precision of $\eps$ in the competitive ratio, the strategies are combined with
a precision of $\eps/2$ in its parameters. 
For any $0< \eps \leq 1$, let ${\cal A}_3^*(\eps)$ be the algorithm obtained by executing
${\cal A}_3(\eps/2)$ in ${\cal A^*}(\eps/2)$. For any $0< \eps \leq 1$, 
let ${\cal A}_1^*(\eps)$ be the algorithm obtained by executing ${\cal A}_1(\eps/2)$ in ${\cal A^*}(\eps/2)$.

\begin{corollary}\label{cor:2}
${\cal A}_3^*(\eps)$ is a $(4/3+\eps)$-competitive algorithm for MPS and
uses no more than $1/\eps^{O(\log (1/\eps))}$ schedules, for any $0<\eps \leq 1$.
\end{corollary}
\begin{proof}
Theorem~\ref{th:guess1} and Corollary~\ref{cor:A3} imply that ${\cal A}_3^*(\eps)$ is
$(4/3+\eps)$-competitive, for any $0<\eps \leq 1$, and that the total number of schedules is the 
product of $1/\eps^{O(\log (1/\eps))}$ and $\lceil \log (1+ 12\rho/\eps) / \log(1+ \eps/(6\rho))\rceil$, 
where $\rho = 4/3+\eps/2$. By the Taylor series for $\ln(1+x)$, $-1< x \leq 1$, we obtain
$\ln(1+x) \geq x/2$, for any $0<x \leq 1$. Hence the second term of the product is $1/\eps^{O(1)}$. \hspace*{\fill}{$\Box$}
\end{proof}

\begin{corollary}\label{cor:3}
${\cal A}_1^*(\eps)$ is a $(1+\eps)$-competitive algorithm for MPS and uses no more than 
$(m/\eps)^{O(\log (1/\eps) / \eps)}$ schedules, for any $0<\eps \leq 1$.
\end{corollary}
\begin{proof}
By Theorems~\ref{th:guess1} and~\ref{th:guess2} algorithm ${\cal A}_1^*(\eps)$ is $(1+\eps)$-competitive,
for any $0<\eps \leq 1$. The total number of schedules is the product of 
$(\lfloor 4m/\eps\rfloor +1)^{\lceil \log(4/\eps) / \log(1+\eps/4) \rceil }$ and
$\lceil \log (1+ 12\rho/\eps) / \log(1+ \eps/(6\rho))\rceil$, where $\rho = 1 +\eps/2$. Again, by the 
Taylor series, $\ln(1+x) \geq x/2$, for any $0<x \leq 1$. Hence both terms of the product are upper bounded
by $(m/\eps)^{O(\log (1/\eps) / \eps)}$. \hspace*{\fill}{$\Box$}
\end{proof}

\section{Lower bounds}\label{sec:lb}
We develop lower bounds that apply  to both \MPS\ and \MPSO. 
\begin{theorem}\label{th:lb1}
Let ${\cal A}$ be a deterministic online algorithm for MPS or MPS$_{\rm opt}$. If ${\cal A}$ attains a competitive 
ratio smaller than $4/3$, then it must maintain at least $\lfloor m/3\rfloor +1$ schedules.
\end{theorem}
\begin{proof}
Let ${\cal A}$ be any deterministic online algorithm for \MPS\ or \MPSO\ that maintains at most $\lfloor m/3\rfloor$
schedules. We show that ${\cal A}$'s competitive ratio is at least $4/3$. To this end we construct an 
adversarial job sequence $\sigma$ such that each schedule maintained by ${\cal A}$ has a makespan
of at least $4/3\cdot \opt(\sigma)$. 

The job sequence $\sigma$ is composed of two subsequences $\sigma_1$ and $\sigma_2$, i.\,e.\
$\sigma = \sigma_1 \sigma_2$. Subsequence $\sigma_1$ consists of $m$ jobs of processing time
$1/3$ each. Subsequence $\sigma_2$ will consist of jobs having a processing time of either
2/3 or 1. The exact number of these jobs depends on the schedules constructed by ${\cal A}$ and 
will be determined later.

Consider the schedules that ${\cal A}$ may have built after all jobs of $\sigma_1$ have been
assigned. Each such schedule contains $m$ jobs of processing time 1/3. For the moment we concentrate on 
schedules in which each machine contains either zero, one or three jobs, i.\,e.\ there exists no machine 
containing two or more than three jobs. Each such schedule $S$ can be represented by a pair $(m_1,m_3)$, 
where $m_1$ denotes the number of machines containing exactly one job and $m_3$ is the number of machines 
containing three jobs. Here $m_1$ and $m_3$ are non-negative integers such that $m_1 + 3m_3 = m$. Let
$P = \{(m_1,m_3) \mid m_1,m_3\in \mathbb{N}_0 \ \mbox{and}\ m_1+3m_3 = m\}$ be the set of all
these pairs. Set $P$ has $\lfloor m/3\rfloor +1$ elements because $m_3$ can take any value
between~0 and $\lfloor m/3\rfloor$ and $m_1 = m - 3m_3$. Let $S$ be an arbitrary schedule 
containing $m$ jobs of processing time 1/3 and $(m_1,m_3)\in P$. We say that $S$ is an 
{\em $(m_1,m_3)$-schedule} if the number of machines containing exactly one job equals $m_1$ and 
the number of machines containing exactly three jobs equals $m_3$. 

Let ${\cal S}$ be the set of schedules constructed by ${\cal A}$ when the entire subsequence 
$\sigma_1$ has arrived. By assumption ${\cal A}$ maintains at most $\lfloor m/3\rfloor$ schedules, 
i.\,e.\ $|{\cal S}| \leq \lfloor m/3\rfloor$. Hence there must exist a pair $(m_1^*,m_3^*)\in P$
such that no schedule of ${\cal S}$ is an $(m_1^*,m_3^*)$-schedule. On the other hand, let ${\cal S}^*$ 
be an $(m_1^*,m_3^*)$-schedule. In ${\cal S}^*$ we number the machines in order of non-decreasing load
such that $\ell^*(1) \leq \ldots \leq \ell^*(m)$. Schedule ${\cal S}^*$ contains $m - m_3^*$ machines
with a load smaller than~1 and, in particular, $m-m_1^*-m_3^*$ empty machines. 

Now the subsequence $\sigma_2$ consists of $m-m_3^*$ jobs, where the $j$-th job has a processing time of
$1-\ell^*(j)$, for $j=1,\ldots, m-m_3^*$. Hence $\sigma_2$ contains $m-m_1^*-m_3^*$ jobs of
processing time~1 followed by $m_1^*$ jobs of processing time~$2/3$. Obviously, the makespan of an 
optimal schedule for $\sigma$ is~1: The jobs of $\sigma_1$ are sequenced so that an $(m_1^*,m_3^*)$-schedule 
is obtained. Again, after $\sigma_1$ has arrived, the machines are numbered in order of non-decreasing
load. While $\sigma_2$ arrives, the $j$-th job is assigned to machine $M_j$, having a load of $\ell^*(j)$,
for $j=1,\ldots, m-m_3^*$. 

In the remainder of this proof we consider any schedule $S\in {\cal S}$ and show that after $\sigma_2$ 
has been sequenced, the resulting makespan is at least 4/3. This establishes the theorem.
So let $S \in {\cal S}$ be any schedule and recall that $S$ contains $m$ jobs of processing time 1/3
each. If in $S$ there exists a machine that contains at least four of these jobs, then the makespan
is already 4/3 and there is nothing to show. Therefore, we restrict ourselves to the case that
every machine in $S$ contains at most three jobs. Again we number the machines in $S$ in order of
non-decreasing load so that $\ell(1) \leq \ldots \leq \ell(m)$. Consider the $(m_1^*,m_3^*)$-schedule
${\cal S}^*$ in which the machines loads satisfy $\ell^*(1) \leq \ldots \leq \ell^*(m)$. There must
exist a machine $M_{j_0}$, $1\leq j_0 \leq m$, such that $\ell(j_0) > \ell^*(j_0)$: For, if  
$\ell(j_0) \leq \ell^*(j_0)$ held for all $j= 1,\ldots, m$, then $\ell(j_0) = \ell^*(j_0)$ for all 
$j= 1,\ldots, m$ because $S$ and $S^*$ both contain jobs with a total processing time of $m/3$. Thus $S$ would
be an $(m_1^*,m_3^*)$-schedule and we obtain a contradiction. The last $m_3^*$ machines in $S^*$ 
have a load of~1. It follows that $j_0 \leq m-m_3^*$ because otherwise $M_{j_0}$ in $S$ contained
at least four jobs. The property $\ell(j_0) > \ell^*(j_0)$ implies $\ell(j_0) \geq \ell^*(j_0) +1/3$
because $S$ and $S^*$ only contain jobs of processing time $1/3$.

We finally show that sequencing of $\sigma_2$ leads to a makespan of at least $4/3$ in $S$. If
${\cal A}$ assigns two jobs of $\sigma_2$ to the same machine, then the resulting machine load is
at least 4/3 because each job of $\sigma_2$ has a processing time of at least $2/3$. So assume that
${\cal A}$ assigns the jobs of $\sigma_2$ to different machines. The first $j_0$ jobs of $\sigma_2$ each
have a processing time of at least $1-\ell^*(j_0)$ because the jobs arrive in order of non-increasing
processing times. In $S$ there exist at most $j_0-1$ machines having a load strictly smaller than
$\ell(j_0)$. Hence, after the first $j_0$ jobs have been scheduled in $S$, there exists a machine
having a load of at least $\ell(j_0) +1 - \ell^*(j_0) \geq \ell^*(j_0) +1/3+1 - \ell^*(j_0) = 4/3$.
This concludes the proof. \hspace*{\fill}{$\Box$}
\end{proof}

The next theorem gives a lower bound on the number of schedules required by a $(1+\eps)$-competitive algorithm,
where $0<\eps < 1/4$. It implies that, for any fixed $\eps$, the number asymptotically depends on 
$m^{\Omega(1/\eps)}$, as $m$ increases. For instance, any algorithm with a competitive ratio
smaller than $1+{1/ 12}$ requires $\Omega(m^2)$ schedules.  Any algorithm with a competitive ratio
smaller than $1+{1/ 16}$ needs $\Omega(m^3)$ schedules.
\begin{theorem}\label{th:lb2}
Let ${\cal A}$ be a deterministic online algorithm for MPS or MPS$_{\rm opt}$. If ${\cal A}$ attains a competitive 
ratio smaller than $1+\eps$, where $0<\eps \leq 1/4$, then it must maintain at least ${m'+h-1 \choose h-1}$
schedules, where $m' = \lfloor m/2 \rfloor$ and $h = \lfloor 1/(4\eps)\rfloor$. The binomial coefficient
increases as $\eps$ decreases and is at least $\Omega((\eps m)^{\lfloor 1/(4\eps)\rfloor-1/2}/\sqrt{m})$.
\end{theorem}
\begin{proof}
We extend the proof of Theorem~\ref{th:lb1}. Let $0<\eps \leq 1/4$. Furthermore, let $m'$ and $h$ be defined as
in the theorem. There holds $h\geq 1$. Let $\eps' = 1/(4h)$ and note that $\eps' \geq \eps$. We will define a set
$M$ whose cardinality is at least ${m'+h-1 \choose h-1}$, and show that if ${\cal A}$ maintains less than
$|M|$ schedules, then its competitive ratio is at least $1+\eps'$.

We specify a job sequence $\sigma$ and first assume that $m$ is even. Later we will describe how to adapt
$\sigma$ if $m$ is odd. Again $\sigma$ is composed of two partial sequences $\sigma_1$ and $\sigma_2$ so
that $\sigma = \sigma_1\sigma_2$. Subsequence $\sigma_1$ consists of $mh$ jobs of processing time $\eps'$ 
each. Subsequence $\sigma_2$ depends on the schedules constructed by ${\cal A}$ and will be specified below. 
Consider the possible schedules after $\sigma_1$ has been sequenced on the $m$ machines. We restrict
ourselves to schedules having the following property: Each machine has a load of exactly~1 or a load
that is at most $1/2 - \eps'$.  Observe that each machine of load~1 contains $1/\eps'$ jobs. Each machine
of load at most $1/2 - \eps'$ contains up to $2h-1$ jobs because $(2h-1)\eps' = 2h/(4h) -\eps' = 
1/2 -\eps'$. Therefore any schedule with the stated property can be described by a vector
$\vec{m} = (m_0, \ldots, m_{2h})$, where $m_{2h}$ is the number of machines having a load of~1 and
$m_i$ is the number of machines containing exactly $i$ jobs, for $i=0, \ldots, 2h-1$. The vector $\vec{m}$
satisfies $\sum_{i=0}^{2h} m_i = m $ and $(1/\eps')m_{2h} + \sum_{i=1}^{2h-1} i m_i = mh$. The last
equation specifies the constraint that the schedule contains $mh$ jobs. Let $M$ be the set of all these
vectors, i.\,e.\
\begin{eqnarray*}
\lefteqn{\textstyle{M = \{(m_0, \ldots, m_{2h}) \in \mathbb{N}_0^{2h+1}\mid  \sum_{i=0}^{2h} m_i = m \ \ \mbox{and} }}\\
& & \hspace*{4.6cm}\textstyle{(1/\eps')m_{2h} + \sum_{i=1}^{2h-1} i m_i = mh\}.}
\end{eqnarray*}
We remark that each $\vec{m}\in M$ uniquely identifies one schedule with our desired property. Let $S$ be
any schedule containing exactly $mh$ jobs of processing time $\eps'$ and $\vec{m} = (m_0, \ldots, m_{2h}) \in M$.
We say that $S$ is an {\em $\vec{m}$-schedule\/} if in $S$ there exist $m_{2h}$ machines of load~1 and 
$m_i$ machines containing exactly $i$ jobs, for $i=0, \ldots, 2h-1$. 

Now suppose that ${\cal A}$ maintains less than $|M|$ schedules. Let ${\cal S}$ be the set of schedules constructed
by ${\cal A}$ after all jobs of $\sigma_1$ have arrived. Since $|{\cal S}| <|M|$ there must exist an 
$\vec{m}^* = (m_0^*, \ldots, m_{2h}^*) \in M$ such that no schedule of ${\cal S}$ is an $\vec{m}^*$-schedule.
Let $S^*$ be an $\vec{m}^*$-schedule in which machines are numbered in order of non-decreasing load such that
$\ell^*(1) \leq \ldots \leq \ell^*(m)$. Subsequence $\sigma_2$ consists of $m-m_{2h}^*$ jobs, where job $j$
has a processing time of $1-\ell^*(j)$, for $j=1, \ldots,m-m_{2h}^*$. Hence $\sigma_2$ consists of $m_i^*$ jobs
of processing time $1-i\eps'$, for $i=0, \ldots, 2h-1$. These jobs arrive in order of non-increasing processing
time. Each job has a processing time of at least $1/2+\eps'$ because $1-(2h-1)\eps' = 1 - (2h/4h-\eps') = 1/2+\eps'$. 
The makespan of an optimal schedule for $\sigma$ is 1. The jobs of $\sigma_1$ are sequenced so that an 
$\vec{m}^*$-schedule is obtained. Machines are again numbered in order of non-decreasing load. Then, while
the jobs of $\sigma_2$ arrive, the $j$-th job of the subsequence is assigned to machine $M_j$ in $S^*$,
$1\leq j \leq m-m_{2h}^*$.

We next show that after ${\cal A}$ has sequenced $\sigma_2$, each of its schedules has a makepan of at least
$1+\eps'$. So consider any $S\in {\cal S}$ and, as always, number the machines in order of non-decreasing load
such that $\ell(1)\leq \ldots \leq \ell(m)$. If in $S$ there exists a machine that has a load of at least 
$1+\eps'$ and hence contains at least $1/\eps'+1$ jobs, then there is nothing to show. So assume that each
machine in $S$ contains at most $1/\eps'$ jobs and thus has a load of at most~1. We study the assignment
of the jobs of $\sigma_2$ to $S$. If ${\cal A}$ places two jobs of $\sigma_2$ on the same machine, then we are
done because each job has a processing time of at least $1/2+\eps'$. Therefore we focus on the case that
${\cal A}$ assigns the jobs of $\sigma_2$ to different machines. 

Schedules $S$ and $S^*$ both contain jobs of total processing time $mh\eps'$. Since $S$ is not an 
$\vec{m}^*$-schedule there must exist a $j_0$, $1\leq j_0\leq m$, such that $\ell(j_0) > \ell^*(j_0)$
and hence $\ell(j_0) \geq \ell^*(j_0)+\eps'$. Each machine in $S$ has a load of at most~1 while the
last $m-m_{2h}^*$ machines in $S^*$ have a load of exactly~1. This implies $j_0\leq m-m_{2h}^*$.
The first $j_0$ jobs of $\sigma_2$ each have a processing time of at least $1-\ell^*(j_0)$. However,
there exist at most $j_0-1$ machines in $S$ having a load strictly smaller than $\ell^*(j_0)$. Hence
after ${\cal A}$ has sequenced the first $j_0$ jobs of $\sigma_2$ there must exist a machine in
$S$ with a load of at least $\ell(j_0) +1 - \ell^*(j_0) \geq \ell^*(j_0) +\eps' + 1 -\ell^*(j_0)
= 1+\eps'$. 

So far we have assumed that $m$ is even. If $m$ is odd, we can easily modify $\sigma$. The first job 
of $\sigma$ is a job of processing time~1. Then $\sigma_1$ and $\sigma_2$ follow. These subsequences
are defined as above, where $m$ is replaced by the even number $m-1$. In this case 
\begin{eqnarray*}
\lefteqn{\textstyle{M = \{(m_0, \ldots, m_{2h}) \in \mathbb{N}_0^{2h-1} \mid  \sum_{i=0}^{2h} m_i = m-1 \ \ \mbox{and}}}\\
& & \hspace{4.6cm} \textstyle{(1/\eps')m_{2h} + \sum_{i=1}^{2h-1} i m_i = (m-1)h\}.}
\end{eqnarray*}
The analysis presented above carries over because the first job of $\sigma$, having a processing time of 1,
must be scheduled on a separate machine and cannot be combined with any job of $\sigma_1$ or $\sigma_2$ if a competitive 
ratio smaller than $1+\eps'$ is to be attained. 

We next lower bound the cardinality of $M$. Again we first focus on the case that $m$ is even. In the definition of
$M$ the critical constraint is $(1/\eps')m_{2h} + \sum_{i=1}^{2h-1} i m_i = mh$, which implies that not every
vector of $\{0, \ldots, m\}^{2h+1}$ represents a schedule that can be built of $mh$ jobs. In particular, the 
vector $(0, \ldots, 0,m)$ of length $2h+1$ would require $m/\eps' = 4h$ jobs. Therefore, we introduce a set
$M'$ and show $|M'| \leq |M|$. Set $M'$ contains vectors of length $2h+1$ in which the first $h+1$ entries
as well as the last one are equal to~0. The other entries sum to at most $m/2$, i.\,e.\
$$\textstyle{M' = \{(0,\ldots, 0,m'_{h+1}, \ldots, m'_{2h-1},0) \in \mathbb{N}_0^{2h+1} \mid \sum_{i=1}^{h-1} m'_{h+i} \leq m/2\}.}$$
We show that each $\vec{m}'\in M'$ can be mapped to a $\vec{m}\in M$. The mapping has the property that
any two different vectors of $M'$ are mapped to different vectors of $M$. This implies $|M'| \leq |M|$.

Consider any $\vec{m}' = (0,\ldots, 0,m'_{h+1}, \ldots, m'_{2h-1},0) \in M'$. Let $\vec{m} = (m_0, \ldots, m_{2h})$
be defined as follows. For $i= h+1, \ldots, 2h$, let $m_i = m'_i$. For $i=0, \ldots, h-1$, let $m_i = m_{2h-i}$.
Finally, let $m_h = m -2\sum_{i=1}^{h-1} m_i$. Note that $m_0=m_{2h} =0$. We next show that $\vec{m}\in M$. 
There holds $\sum_{i=0}^{2h} m_i = \sum_{i=1}^{2h-1} m_i = 2\sum_{i=1}^{h-1} m_i +m_h = m$. Furthermore,
\begin{eqnarray*}
m_{2h}/\eps' + \sum_{i=0}^{2h-1} im_i & = & \sum_{i=1}^{2h-1} im_i = \sum_{i=1}^{h-1} (i+2h-i) m_i +hm_h \\
& = & 2h\sum_{i=1}^{h-1}m_i + h(m-2\sum_{i=1}^{h-1} m_i) = mh.
\end{eqnarray*}
It follows, as desired, $\vec{m}\in M$. Note that the last $h$ entries of $\vec{m}$ are identical to the last $h$ entries 
of $\vec{m}'$. Hence no two vectors of $M'$ that differ in at least one entry are mapped to the same vector of
$M$. Hence $|M'| \leq |M|$. If the number $m$ of machines is odd, then in the definition of $M'$ the entries of a vector 
sum to at most $(m-1)/2$. The rest of the construction and analysis is the same. Thus, for a general number $m$ of machines
$$\textstyle{M' = \{(0,\ldots, 0,m'_{h+1}, \ldots, m'_{2h-1},0) \mid  m'_i\in \mathbb{N}_0 \ \ \mbox{and} \ \
\sum_{i=1}^{h-1} m'_{h+i} \leq \lfloor m/2\rfloor\}.}$$
This set contains exactly ${m'+h-1 \choose h-1}$ elements, where again $m' = \lfloor m/2\rfloor$. In the remainder
of this proof we lower bound this binomial coefficient. 

There holds
$\sqrt{2\pi e} (k/e)^{k+1/2} \leq  k! \leq 2 \sqrt{2\pi e} (k/e)^{k+1/2}$
 for any $k\in \mathbb{N}$ by Stir\-ling's approximation~\cite{F}. Hence
\begin{eqnarray*}
{m'+h-1\choose h-1} &=& {(m'+h-1)!\over m'!(h-1)!} \  \geq {(m'+h-1)^{m'+h-1/2}\over 4\sqrt{2\pi} (m')^{m'+1/2} (h-1)^{h-1/2}}\\
&=& {1\over 4\sqrt{2\pi m'}} \left(1+{h-1\over m'}\right)^{m'} \left(1+{m'\over h-1}\right)^{h-1/2} \\
& > & {1\over 4\sqrt{2\pi m'}} \left(1+{m/2-1/2 \over 1/(4\eps)}\right)^{h-1/2}.
\end{eqnarray*}
The last expression is $\Omega((\eps m)^{\lfloor 1/(4\eps)\rfloor -1/2}/\sqrt{m})$.\hspace*{\fill}{$\Box$}
\end{proof}

\end{document}